\documentclass{article}
\usepackage{natbib}

\usepackage[final]{neurips_2023}

\usepackage[utf8]{inputenc} 
\usepackage[T1]{fontenc}    
\usepackage{hyperref}       
\usepackage{url}            
\usepackage{booktabs}       
\usepackage{amsfonts}       
\usepackage{nicefrac}       
\usepackage{microtype}      
\usepackage{xcolor}         

\usepackage{amsthm}
\usepackage{amsmath}
\usepackage{amssymb}
\usepackage{mathtools}
\usepackage{graphicx}
\usepackage{cleveref}
\usepackage{enumerate}
\newcommand\numberthis{\addtocounter{equation}{1}\tag{\theequation}}
\newtheorem{theorem}{Theorem}

\newtheorem{proposition}[theorem]{Proposition}
\newtheorem{lemma}[theorem]{Lemma}

\newtheorem{definition}[theorem]{Definition}

\newif\ifnotes\notestrue

\ifnotes

\newcommand{\YT}[1]{{\textbf{\color{blue}YT:~ #1}}}

\else
\newcommand{\YT}[1]{ #1}

\fi

\usepackage{mwe} 
\usepackage{wrapfig}

\usepackage{tikz}

\newcommand{\hidecontent}[1]{}

\newcommand{\rev}[1]{\texttt{Rev}({#1})}
\newcommand{\revsoftmax}[1]{\texttt{Rev}^{\texttt{softmax}}({#1})}

\newcommand{\xx}{\lceil \frac{4}{\epsilon^2}\rceil^n}
\newcommand{\xxs}{\lceil \frac{4}{\epsilon^2}\rceil^{2n}}

\newcommand{\yys}{\lceil \frac{16m^3}{\epsilon^2}\rceil^{2nm}}
\newcommand{\BW}[1]{\text{BW}(#1)}

\usetikzlibrary{arrows.meta, shapes}
\tikzset{bigneuron/.style={circle, draw, inner sep = 0, minimum width = 5.5ex}}
\tikzset{smallneuron/.style={circle, draw, inner sep = 0, minimum width = 4ex}}
\tikzset{transform/.style={fill=white, circle}}
\tikzset{connection/.style={-{Stealth}}}

\newcommand{\R}{\mathbb{R}}
\newcommand{\options}{\mathcal{K}}
\newcommand{\boost}{\beta}
\newcommand{\menus}{\mathcal{M}}

\title{Mode Connectivity in Auction Design}

\author{%
  Christoph Hertrich\thanks{Moved to Université Libre de Bruxelles, Belgium, and Goethe-Universit\"at Frankfurt, Germany, after submission of this article.} \\
  Department of Mathematics\\
  London School of Economics and Political Science, UK\\
  \texttt{c.hertrich@lse.ac.uk}
  \And
  Yixin Tao \\
  ITCS, Key Laboratory of Interdisciplinary Research of Computation and Economics\\
   Shanghai University of Finance and Economics, China\\
  \texttt{taoyixin@mail.shufe.edu.cn}
  \And
  László A.\ Végh \\
  Department of Mathematics\\
  London School of Economics and Political Science, UK\\
  \texttt{l.vegh@lse.ac.uk}
}

\begin{document}

\maketitle

\begin{abstract}
	Optimal auction design is a fundamental problem in algorithmic game theory. This problem is notoriously difficult already in very simple settings. Recent work in differentiable economics showed that neural networks can efficiently learn known optimal auction mechanisms and discover interesting new ones. In an attempt to theoretically justify their empirical success, we focus on one of the first such networks, \emph{RochetNet}, and a generalized version  for \emph{affine maximizer auctions}. We prove that they satisfy \emph{mode connectivity}, i.e., locally optimal solutions are connected by a simple, piecewise linear path such that every solution on the path is almost as good as one of the two local optima. Mode connectivity has been recently investigated as an intriguing empirical and theoretically justifiable property of neural networks used for prediction problems. Our results give the first such analysis in the context of differentiable economics, where neural networks are used directly for solving non-convex optimization problems.
\end{abstract}


\section{Introduction}
Auction design  is a core problem in mechanism design, with immense applications in electronic commerce (such as sponsored search auctions) as well as in the public sector (such as spectrum auctions). In a revenue maximizing auction, the auctioneer needs to design a mechanism to allocate resources to buyers, and set prices in order to maximize the expected revenue. The buyers' preferences are private and they may behave strategically by misreporting them. For this reason, it is often desirable to devise \emph{dominant strategy incentive compatible (DSIC)} and \emph{individually rational (IR)} mechanisms. By definition, in a DSIC mechanism, it is a dominant strategy for the buyers to report the true valuations; in an IR mechanism, each participating truthful buyer receives a nonnegative payoff.

We focus on DSIC and IR mechanisms that maximize the expected revenue, assuming that the buyers' preferences are drawn from a distribution known to the auctioneer. A classical result of \citet{myerson1981optimal} provides the optimal mechanism for the case of a single item and arbitrary number of buyers. Finding the optimal mechanisms for more general settings is a tantalizingly difficult problem. We refer the reader to the surveys by \citet{rochet2003economics,manelli2007multidimensional} and \citet{daskalakis2015multi} for partial results and references. In particular, no analytic solution is known even for two items and two buyers. Selling multiple items to a single buyer is computationally intractable \citep{daskalakis2014complexity}. Already for two items and a single buyer, the description of the optimal mechanism may be uncountable \citep{daskalakis2013mechanism}. Recent work gives a number of important partial characterizations, e.g. \citet{daskalakis2015strong,giannakopoulos2014duality}, as well as results for weaker notions of Bayesian incentive compatibility, e.g. \citet{cai2012optimal, cai2012algorithmic, cai2013understanding, bhalgat2013optimal}. 

\citet{conitzer2002complexity,conitzer2004self} proposed the approach of \emph{automated mechanism design} to use optimization and computational methods to obtain (near) optimal mechanisms for specific problems; see also \citet{sandholm2015automated}. An active recent area of research uses machine learning tools. In particular, \citet{dutting2019optimal} designed and trained neural networks to automatically find optimal auctions. They studied two network architectures, and showed that several theoretically optimal mechanisms can be recovered using this approach, as well as interesting new mechanisms can be obtained. The first network they studied is \emph{RochetNet}. This is  a simple two-layer neural network applicable to the single buyer case,  leveraging \citeauthor{rochet1987necessary}'s [\citeyear{rochet1987necessary}] characterization of the optimal mechanism. The second network, \emph{RegretNet} does not require such a characterization and is applicable for multiple buyers; however, it only provides approximate incentive compatibility. 

\citet{dutting2019optimal} coined the term  \emph{`differentiable  economics'} for this approach, and there has been significant further work in this direction. These include designing auctions for budget constrained buyers \citep{feng2018deep}; multi-facility location \cite{golowich2018deep}; balancing fairness and revenue objectives \citep{kuo2020proportionnet}; incorporating non-linear utility functions and other networks trained from interaction data \citep{shen2019automated}\footnote{MenuNet, developed by \citep{shen2019automated}, also encodes menu items as RochetNet. However, unlike RochetNet’s approach of repeatedly sampling valuations from the underlying distribution, MenuNet discretizes the buyer’s valuation space into discrete values.}; designing revenue-maximizing auctions with
differentiable matchings \citep{curry2022learning}; contextual auction design \citep{duan2022context}; designing taxation policies \citep{zheng2021ai}, and more.

The purpose of this work is to supply theoretical evidence behind the success of neural networks in differentiable economics. The revenue is a highly non-convex function of the parameters in the neural network.  Curiously,  gradient approaches seem to recover globally optimal auctions despite this non-convexity. Similar phenomena have been studied more generally in the context of deep networks, and theoretical explanations have been proposed, in particular, overparametrization \citep{allen2019convergence,du2019gradient}.

\paragraph{Mode connectivity} Recent work has focused on a striking property of the landscape of loss functions of deep neural networks: local optimal solutions (modes) found by gradient approaches are connected by simple paths in the parameter space. We provide an informal definition of \emph{mode connectivity} here.
\begin{definition}[$\varepsilon$-mode connected (informal)]We say that two solutions are \emph{$\varepsilon$-mode connected}, if they are connected by a continuous path of solutions, such that the loss function does not worsen by more than $\varepsilon$ compared to one of the end points on the entire path.
\end{definition}

This phenomenon was identified by \citet{garipov2018loss} and by \citet{draxler2018essentially}. Mode connectivity can help to explain the empirical performance of stochastic gradient descent (sgd) (or ascent, in case of revenue maximization). To some extent, mode connectivity prevents a poor local minimal valley region on the function value, from which the sgd method cannot escape easily. Suppose such a bad local minimum exists. Then mode connectivity implies that there exists a path from this bad local minimum to a global minimum on which the loss function does not significantly increase. Therefore, the intuition is that from every bad local minimum, a (stochastic) gradient method would eventually be able to find a way to escape.
However, we would like to emphasize that mode connectivity does not provide a formal proof of the success of sgd. It only provides a useful intuition of why sgd does not completely trapped in local optima.

\citet{kuditipudi2019explaining} gave strong theoretical arguments for mode connectivity. They introduce the notion of \emph{$\varepsilon$-dropout stability}: solutions to a neural network such that in each layer, one can remove at least half the neurons and rescale the remaining units such that the loss function increases by at most $\varepsilon$. Solutions that are 
$\varepsilon$-dropout stable are then shown to be $\varepsilon$-mode connected.   Moreover, they show that \emph{noise stability} (see e.g., \citet{arora2018stronger}) implies dropout stability, and hence, mode connectivity. \citet{nguyen2019connected} showed mode connectivity when there is a hidden layer larger than the training dataset. \citet{shevchenko2020landscape}  shows that stochastic gradient descent solutions to sufficiently overparametrized neural networks are dropout stable even if we only keep a small, randomly sampled set of neurons from each layer. 

\subsection{Our contributions} 
\paragraph{RochetNet} In this paper, we first establish mode connectivity properties of \emph{RochetNet}, the architecture in \citet{dutting2019optimal} for multiple items and a single buyer. These networks have a single hidden layer, corresponding to the menu. Within this hidden layer, each neuron directly corresponds to an \emph{option} in the menu: which contains an allocation and a price offered to the buyer (see Figure~\ref{fig:rochet}). The buyer is assigned the single option, including the allocation and the price, maximizing the buyer's utility. Such an option is called \emph{active} for the buyer. The loss function of \emph{RochetNet} is the revenue, the expected price paid by the buyer. Despite its simplicity, the  experiments on \emph{RochetNet} in \citet{dutting2019optimal} gave impressive empirical results in different scenarios. For example, in the experiments with up to six items and uniform value distributions, \emph{RochetNet} achieves almost the same revenue ($99.9\%$) as the Straight-Jacket Auctions in \citet{giannakopoulos2018duality}, which are known to be optimal in this case. This success is not limited to a single example, as \emph{RochetNet} also consistently performs well in other scenarios, including when infinite menu size is necessary.  Furthermore, \citet{dutting2019optimal} demonstrated the usefulness of \emph{RochetNet} in discovering optimal auctions in situations that were previously unexplored from a theoretical perspective. 

 First, in Theorem~\ref{thm:epsilon-reducible}, we show  that for linear utilities, $\varepsilon$-mode connectivity holds between two solutions that are \emph{$\varepsilon$-reducible}: out of the $K+1$ menu options (neurons), there exists a subset of at most $\sqrt{K+1}$ containing an active option for the buyer with probability at least $1-\varepsilon$. Assuming that the valuations are normalized such that the maximum valuation of any buyer is at most one, it follows that if we remove all other options from the menu, at most $\varepsilon$ of the expected revenue is lost. The assumption of being $\varepsilon$-reducible is stronger than $\varepsilon$-dropout stability that only drops a constant fraction of the neurons. At the same time, experimental results in \citet{dutting2019optimal} show evidence {of this property being satisfied in practice.} They experimented with different sized neural networks in a setting when the optimal auction requires infinite menu size.
Even with $10,000$ neurons available, only $59$ options were active, i.e., used at least once when tested over a large sample size.
We note that this property also highlights an advantage of 
 \emph{RochetNet} over \emph{RegretNet} and other similar architectures: instead of a black-box neural network, it returns a compact, easy to understand representation of a mechanism.

Our second main result {(Theorem~\ref{thm:large})} shows that for $n$ items and linear utilities, if the number of menu options $K$ is sufficiently large, namely, {$(2/\varepsilon)^{4n}$}, then $\varepsilon$-mode connectivity holds between \emph{any} two solutions for \emph{any} underlying distribution. The connectivity property holds pointwise: for any particular valuation profile, the revenue may decrease by at most $\varepsilon$ along the path. A key tool in this $\varepsilon$-mode connectivity result is a discretization technique from \citet{dughmi2014sampling}.
We note that such a mode connectivity result can be expected to need a large menu size. In Appendix~\ref{sec:example-non-concavity}, we present an example with two disconnected local maxima for $K=1$.

\paragraph{Affine Maximizer Auctions} We also extend our results and techniques to neural networks for affine maximizer auctions (AMA) studied in \cite{curry2022differentiable}. This is a generalization of \emph{RochetNet} for multi-buyer scenarios. It can also be  seen as a weighted variant of the Vickrey--Clarke--Groves (VCG) mechanism \citep{vickrey1961counterspeculation, clarke1971multipart, groves1973incentives}.  AMA offers various allocation options. For a given valuation profile, the auctioneer chooses the allocation with the highest weighted sum of valuations, and computes individual prices for the buyers; the details are described in Section~\ref{sec:ama}. AMA is DSIC and IR, however, is not rich enough to always represent the optimal auction.

 For AMA networks, we show similar results (Theorem~\ref{thm:AMA-1} and \ref{thm:AMA-2}) as for \emph{RochetNet}. We first prove $\varepsilon$-mode connectivity holds between two solutions that are \emph{$\varepsilon$-reducible} (see Definition~\ref{def:ama-reducible}).  \cite{curry2022differentiable} provides evidence {of this property being satisfied in practice}, observing (in Sec.\ 7.1) \emph{``Moreover, we found that starting out with a large number of parameters improves performance, even though by the end of training only a tiny number of these parameters were actually used.''}. Secondly, we also show that  if the number of menu options $K$ is sufficiently large, namely, $(16m^3/\varepsilon^2)^{2nm}$, then $\varepsilon$-mode connectivity holds pointwise between \emph{any} two solutions. That is, it is valid for any underlying distribution of valuations, possibly correlated between different buyers.

\paragraph{Relation to previous results on mode connectivity}
Our results do not seem to be deducible from previous mode connectivity results, as we outline as follows.
Previous literature on mode connectivity investigated  neural networks used for prediction. The results in \cite{kuditipudi2019explaining,nguyen2019connected,shevchenko2020landscape} and other papers crucially rely on the properties that the networks minimize a convex loss function between the predicted and actual values, and require linear transformations in the final layer. 
\emph{RochetNet} and AMA networks are fundamentally different. The training data does not come as labeled pairs and these network architectures are built directly for solving an optimization problem. For an input valuation profile, the loss function is the negative revenue of the auctioneer. In \emph{RochetNet}, this is obtained as 
the negative price of the utility-maximizing bundle; for AMA it requires an even more intricate calculation. The {objective} is to find parameters of the neural network such that the expected revenue is as large as possible. The menu options define a piecewise linear surface of utilities, and the revenue in  \emph{RochetNet} can be interpreted as the expected bias of the piece corresponding to a randomly chosen input.

Hence, the landscape of the  loss function is fundamentally different from those analyzed in the above mentioned works. The  weight interpolation argument that shows mode-connectivity from dropout stability is not applicable in this context. The main reason is that the loss function is not a simple function of the output of the network, but is defined by choosing the price of the argmax option.
We thus need a more careful understanding of the piecewise linear surfaces corresponding to the menus.

\paragraph{Significance for Practitioners} We see the main contribution of our paper in \emph{explaining} the empirical success and providing theoretical foundations for already existent practical methods, and not in inventing new methods. Nevertheless, two insights a practitioner could use are as follows: (i) It is worth understanding the structure of the auction in question. If one can, e.g., understand whether $\varepsilon$-reducibility holds for a particular auction, this might indicate whether RochetNet or AMA are good methods to apply to this particular case. (ii) Size helps: If one encounters bad local optima, increasing the menu size and rerunning RochetNet or AMA might be a potential fix and will eventually lead to a network satisfying mode connectivity.


\section{Auction Settings}
We consider the case with $m$ buyers and one seller with $n$ divisible items each in unit supply. Each buyer has an additive valuation function $v_i(S) := \sum_{j \in S} v_{ij}$, where $v_{ij} \in V$ represents the valuation of the buyer $i$ on item $j$ and $V$ is the set of possible valuations. Throughout the paper, we normalize the range to the unit simplex:  we assume $V = [0, 1]$, and  $\|v_i\|_1 = \sum_j v_{ij} \leq 1$ for every buyer $i$. With slight abuse of notation, we let $v = (v_{11}, v_{12}, \cdots, v_{ij}, \cdots, v_{mn})^\top$ and $v_i = (v_{i1}, v_{i2}, \cdots, v_{in})^\top$. The buyers' valuation profile $v$ is drawn from a distribution $F\in \mathcal{P}(V^{m\times n})$. 
Throughout, we assume that the buyers have \emph{quasi-linear utilities}: if a buyer with valuation $v_i$ receives an allocation $x\in [0,1]^n$ at price $p$, their utility is $v_i^\top x-p$.

The seller has access to samples from the distribution $F$, and wants to sell these items to the buyers through a DSIC\footnote{There exists a weaker notion of incentive compatibility: \emph{Bayesian incentive compatible} (BIC). In a BIC mechanism, it is a dominant strategy for the buyers to report the true valuations if all other buyers also report truthfully. This paper focuses on DISC mechanisms, as DSIC is considered to be more robust than BIC, which assumes a common knowledge of buyers' distributions on valuations.} and IR auction and maximize the expected revenue.  
In the auction mechanism, the $i$-th bidder reports a \emph{bid} $b_i\in [0,1]^n$. The entire bid vector $b\in [0,1]^{m\times n}$ will be denoted as $b=(b_1,\ldots,b_m)=(b_i,b_{-i})$, where $b_{-i}$ represents all the bids other than buyer $i$. In a DSIC mechanism, it is a dominant strategy for the agents to report $b_i=v_i$, i.e., reveal their true preferences.

\begin{definition}[DSIC and IR auction]
An auction mechanism requires the buyers to submit bids $b_i\in [0,1]^n$, and let $b=(b_1,\ldots,b_m)$. The output is a set of allocations $x(b)=(x_1(b),\ldots,x_m(b))$, $x_i(b)\in [0,1]^n$, and prices $p(b)=(p_1(b),\ldots,p_m(b))\in \mathbb{R}^m$.
Since there is unit supply of each item, we require $\sum_i x_{ij}(b) \leq 1$, where $x_{ij}(b)$ is the allocation of buyer $i$ of item $j$.
\begin{enumerate}[(i)]
\item  An auction is \emph{dominant strategy incentive compatible (DSIC)} if $v_i^\top x_i(v_i, b_{-i}) - p_i(v_{i}, b_{-i}) \geq v^\top x_i(b_i, b_{-i}) - p_i(b_i, b_{-i})$ for any buyer $i$ and any bid $b=(b_i, b_{-i})$.
\item 
 An auction is \emph{individually rational (IR)} if $v_i^\top x_i(v_i, b_{-i}) - p_i(v_i, b_{-i}) \geq 0$.
\end{enumerate}
\end{definition}
The revenue of a DSIC and IR auction is
\begin{align*}
    \texttt{Rev} = \mathbb{E}_{v\sim F} \left[ \sum_i p_i(v)\right].
\end{align*}

\subsection{Single Buyer Auctions: RochetNet}
\citet{dutting2019optimal} proposed RochetNet as a DISC and IR auction for the case of a single buyer. We omit the subscript $i$ for buyers in this case. 
A (possibly infinite sized) \emph{menu} $M$ comprises a set of \emph{options} offered to the buyer: $M = \{(x^{(k)}, p^{(k)})\}_{k\in \options}$. In each option $(x^{(k)}, p^{(k)})$, $x^{(k)}\in [0,1]^n$ represents the amount of items, and $p^{(k)}\in \R_+$ represents the price. We assume that $0\in \options$, and $(x^{(0)}, p^{(0)})=(\mathbf{0},0)$ to guarantee IR. We call this the \emph{default option}, whereas all other options are called \emph{regular options}. We will use $K$ to denote the number of regular options; thus, $|\options|=K+1$.
 
A buyer submits a bid $b\in [0,1]^n$ representing their valuation, and is assigned to option $k(b)\in\options$ that maximizes the utility\footnote{We assume that ties are broken in favor of higher prices, but it is not hard to see that our results transfer to other tie-breaking rules, too. See also the discussion in Section~B.2 of \citet{babaioff2022menu}.}
\begin{align*}
k(b)\in \arg\max_{k\in\options}
    b^\top x^{(k)} - p^{(k)} \, .
\end{align*} 
This is called the \emph{active option} for the buyer.
Note that option $0$ guarantees that the utility is nonnegative, implying the IR property. It is also easy to see that such an auction is DSIC. Therefore, one can assume that $b=v$, i.e., the buyer submits their true valuation; or equivalently, the buyer is allowed to directly choose among the menu options one that maximizes their utility.
Moreover, it follows from \cite{rochet1987necessary} that every DSIC and IR auction for a single buyer can be implemented with a (possibly infinite size) menu {using an appropriate} tie-breaking rule.

Given a menu $M$, the revenue is defined as 
\begin{align*}
\rev{M} &=  \mathbb{E}_{v\sim F} \left[ p^{(k(v))} \right]\, .
\end{align*}

\paragraph{RochetNet}
RochetNet (see Figure~\ref{fig:rochet}) is a neural network with three layers: an input layer ($n$ neurons), a middle layer ($K$ neurons), and an output layer ($1$ neuron):
\begin{enumerate}
    \item the input layer takes an $n$-dimensional bid $b \in V^n$, and sends this information to the middle layer;
    \item the middle layer has $K$ neurons. Each neuron represents a regular option in the menu $M$, which has parameters $x^{(k)} \in [0,1]^n$ and $p^{(k)} \in \R_{+}$, where $x^{(k)}\in [0,1]^n$ represents the allocation of option $k$ and $p^{(k)}$ represents the price of option $k$.  Neuron $k$ maps from $b \in V^n$ to $b^\top x^{(k)} - p^{(k)}$, i.e., the utility of the buyer when choosing option $k$; 
    \item the output layer receives all utilities from different options and maximizes over these options and $0$: $\max\{\max_k \{(x^{(k)})^\top b - p^{(k)}\}, 0\}$.
\end{enumerate}

 We will use $\rev{M}$ to denote the  revenue of the auction with menu options $\options=\{0,1,2,\ldots,K\}$, where $0$ represents the default option $(\mathbf{0},0)$.

The training objective for the RochetNet is to maximize the revenue $\rev{M}$, which is done by stochastic gradient ascent. Note, however, that the revenue is the price of an \emph{argmax} option, which makes it a non-continuous function of the valuations. For this reason, \citet{dutting2019optimal} use a \emph{softmax}-approximation of the \emph{argmax} as their loss function instead. However, \emph{argmax} is used  for testing. In Appendix~\ref{sec:softmax-rochet}, we bound the difference between the revenues computed with these two different activation functions, assuming that the probability density function of the distribution $F$ admits a finite upper bound. Lemma~\ref{lem:softmax} shows that the difference between the revenues for \emph{softmax} and \emph{argmax} is roughly inverse proportional to the parameter $Y$ of the \emph{softmax} function. This allows the practitioner to interpolate between smoothness of the loss function and provable quality of the softmax approximation by tuning the parameter $Y$.
  
\begin{figure}
	\centering
	\begin{tikzpicture}
		\footnotesize
		\node[smallneuron, label=below:{$-p_1$}] (n1) at (0,11ex) {};
		\node[smallneuron, label=below:{$-p_2$}] (n2) at (0,2ex) {};
		\node[rotate=90] (npunkt) at (0,-6ex) {$\mathbf{\cdots}$};
		\node[smallneuron, label=below:{$-p_k$}] (n3) at (0,-11ex) {};
		\node[smallneuron, label=left:{$b_1$}] (in1) at (-24ex,7.5ex) {};
		\node[smallneuron, label=left:{$b_2$}] (in2) at (-24ex,2.5ex) {};
		\node[rotate=90] (inpunkt) at (-24ex,-2.5ex) {$\mathbf{\cdots}$};
		\node[smallneuron, label=left:{$b_n$}] (in3) at (-24ex,-7.5ex) {};
		\node[bigneuron] (out) at (24ex,0ex) {$\max$};
		\draw[connection] (in1) -- (n1) node[above, pos=0.8] {$x^{(1)}$};
		\draw[connection] (in2) -- (n1);
		\draw[connection] (in3) -- (n1);
		\draw[connection] (in1) -- (n2);
		\draw[connection] (in2) -- (n2);
		\draw[connection] (in3) -- (n2) node[below, pos=0.8] {$x^{(2)}$};
		\draw[connection] (in1) -- (n3);
		\draw[connection] (in2) -- (n3);
		\draw[connection] (in3) -- (n3) node[below, pos=0.8] {$x^{(k)}$};
		\draw[connection] (n1) -- (out);
		\draw[connection] (n2) -- (out);
		\node[rotate=90] (npunkt) at (12ex,-3ex) {$\mathbf{\cdots}$};
		\draw[connection] (n3) -- (out);
		\node[smallneuron] (zero) at (15ex, 11ex) {$0$};
		\draw[connection] (zero) -- (out);
		\draw[connection] (out) -- (30ex,0ex);
	\end{tikzpicture}
	\caption{RochetNet: this architecture maps the bid $b$ to the utility of the buyer.}
	\label{fig:rochet}
\end{figure}
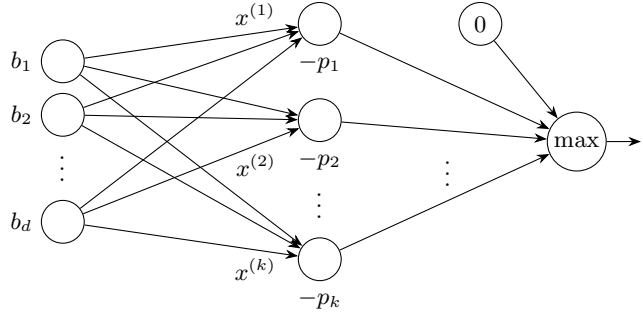

\subsection{Affine Maximizer Auctions}
\label{sec:ama}
\emph{Affine Maximizer Auctions (AMA)}
also provide a menu $M$ with a set of options $\options$.  Each option
is of the form $(x^{(k)},\boost^{(k)})\in [0,1]^{n\times m}\times \R$, where   $x^{(k)}_{ij}\in [0,1]$ represents the allocation of item $i$ to buyer $j$, with the restriction that $\sum_i x^{(k)}_{ij} \leq 1$ for each item $j$, and $\boost^{(k)}$ represents a \emph{`boost'}. We again
assume $0\in \options$, and $(x^{(0)}, \beta^{(0)})=(\mathbf{0},0)$, and call this the \emph{default} option; all other options are called the \emph{regular options}.

Given the bids $b_{i}\in [0,1]^n$ of the agents, 
the auctioneer computes a weighted welfare, using weights $w_i\in \R_{+}$ for the valuations of each agent, and adds the boost $\boost^{(k)}$. Then, the allocation maximizing the weighted boosted welfare is chosen, i.e., the option with 
\begin{align*}
    k(b) \in \arg \max_{k\in{\options}} \sum_i  w_i b_i^\top x_{i}^{(k)} + \boost^{(k)}.
\end{align*}
This will also be referred to as the \emph{active option}.
The prices collected from the buyers are computed according to the Vickrey--Clarke--Groves (VCG) scheme. Namely,
\begin{equation}\label{eq:prices-formula}
\begin{aligned}
    p_i(b) = &\frac{1}{w_i} \left( \sum_{\ell \neq i}  w_\ell b_{\ell}^\top x_{\ell}^{(k(b_{-i}))} + \boost^{(k(b_{-i}))} \right)  - \frac{1}{w_i} \left( \sum_{\ell \neq i}  w_\ell b_{\ell}^\top x_{\ell}^{(k(b))} + \boost^{(k(b))} \right).
\end{aligned}
\end{equation}
Here, $k(b_{-i})$ represents the option maximizing the weighted boosted welfare when buyer $i$ is omitted, i.e.,  $k(b_{-i}) \in \arg \max_{k\in{\options}} \sum_{\ell\neq i}  w_\ell b_\ell^\top x_{\ell}^{(k)} + \boost^{(k)}$. It is known that AMA is DSIC and IR. Hence, we can assume that the submitted bids $b_i$ represent the true valuations $v_i$.
We also assume the  ties are broken in favor of maximizing the total payment.  In case of unit weights, this is equivalent to choosing the smallest $\boost^{(k)}$ values, see \eqref{eq:total-payment} in Section~\ref{sec:AMA}.
Given the menu $M$, the revenue of the AMA is 
\begin{align*}
    \rev{M} = \mathbb{E}_{v\sim F} \left[ \sum_i p_i(v)\right]\, .
\end{align*}

In this paper, we focus on the case when $w_i=1$ for all buyers. This is also used in the experiments in \citet{curry2022differentiable}. 
For this case, AMA can be implemented by a three layer neural network similar to \emph{RochetNet}, with $m\times n$ input neurons. For the more general case when the weights $w_i$ can also be adjusted, one can include an additional layer that combines the buyers' allocations.

Note that for a single buyer and $w_1=1$, AMA corresponds to \emph{RochetNet}, with price $p^{(k)}=-\boost^{(k)}$ for each menu option. Indeed, in the formula defining the price $p_i(b)$, the first term is 0, as well as the sum in the second term.

{Similarly to \emph{RochetNet}, the loss function, which is maximized via stochastic gradient ascent, is a \emph{softmax}-approximation of the revenue $\rev{M}$, in order to avoid the discontinuities introduced by the \emph{argmax}.} We bound the difference in the revenue in Appendix~\ref{sec:softmax-ama}, {concluding that it decreases with large parameter $Y$ as in the RochetNet case.}

\subsection{Mode Connectivity}
One can view the revenue as a function of the menus, i.e., the parameters in the mechanism: {\em (i)} in \emph{RochetNet},  $\{(x^{(k)}, p^{(k)})\}_{k\in\options}$; {\em (ii)} in AMA, $\{(x^{(k)}, \boost^{(k)})\}_{k\in \options}$. We use $\menus$ to denote the set of all possible menus.

\begin{definition}(Mode connectivity)
   Two menus $M_1,M_2\in\menus$ are $\varepsilon$-mode-connected if there is a continuous curve $\pi: [0,1]\to \menus$ such that {\em (i)} $\pi(0) = M_1$; {\em (ii)} $\pi(1) = M_2$; and {\em (iii)} for any $t \in [0, 1]$, $\rev{\pi(t)} \geq \min \{\rev{M_1}, \rev{M_2}\} - \varepsilon$.
\end{definition}


\section{Mode Connectivity for the RochetNet}
\label{sec:rochetnet}

In this section we present and prove our main results for the RochetNet. For some statements, we only include proof sketches. The detailed proofs can be found in \Cref{sec:rochetproofs} in the supplementary material. The following definition plays an analogous role to $\varepsilon$-dropout stability in \cite{kuditipudi2019explaining}.

\begin{definition}
    A menu $M$ with $|\options|=K+1$ options  is called \emph{$\varepsilon$-reducible} if there is a subset $\mathcal{K}'\subseteq\mathcal{K}$ with $0\in\mathcal{K}'$, $|\mathcal{K}'|\le \sqrt{K+1}$ such that, with probability at least $1-\varepsilon$ over the distribution of the valuation of the buyer, the active option assigned to the buyer is contained in $\mathcal{K}'$. 
\end{definition}
As noted in the Introduction, such a property can be observed in the experimental results in \cite{dutting2019optimal}.
The motivation behind this definition is that if a menu satisfies this property, then all but $\sqrt{K+1}$ options are more or less redundant. In fact, if a menu is $\varepsilon$-reducible, then dropping all but the at most $\sqrt{K+1}$ many options in $\mathcal{K}'$ results in a menu $M'$ with $\rev{M'}\ge \rev{M}-\varepsilon$ because the price of any selected option is bounded by $\lVert v\rVert_1\leq 1$.

As a first step towards showing the mode connectivity results, we show that $0$-reducibility implies $0$-mode-connectivity. We will then use this to derive our two main results, namely that two $\varepsilon$-reducible menus are always $\varepsilon$-mode-connected and that two large menus are always $\varepsilon$-mode-connected.

\begin{proposition}\label{prop:zero}
    If two menus $M_1$ and $M_2$ for the RochetNet are $0$-reducible, then they are $0$-mode-connected. Moreover, the curve transforming $M_1$ into $M_2$ is piecewise linear with only three pieces.
\end{proposition}

To prove \Cref{prop:zero}, we introduce two intermediate menus $\widehat M_1$ and $\widehat M_2$, and show that every menu in the piecewise linear interpolation from $M_1$ via $\widehat M_1$ and $\widehat M_2$ to $M_2$ yields a revenue of at least $\min\{\rev{M_1},\rev{M_2}\}$.
Using that menu $M_1$ has only $\sqrt{K+1}$ non-redundant options, menu $\widehat M_1$ will be defined by repeating each of the $\sqrt{K+1}$ options $\sqrt{K+1}$ times. Menu $\widehat M_2$ will be derived from $M_2$ similarly. A technical lemma makes sure that this copying can be done in such a way that each pair of a non-redundant option of $M_1$ and a non-redundant option of $M_2$ occurs exactly for one index in $\widehat M_1$ and $\widehat M_2$.

To make this more formal, we first assume without loss of generality that $K+1$ is a square, such that $\sqrt{K+1}$ is an integer. It is straightforward to verify that the theorem is true for non-squares $K+1$, too. Suppose the options in $M_1$ and $M_2$ are indexed with $k\in\mathcal{K}=\{0,1,\dots,K\}$. Since $M_1$ is $0$-reducible, there is a subset $\mathcal{K}_1\subseteq\mathcal{K}$ with $0\in\mathcal{K}_1$,  $|\mathcal{K}_1|=\sqrt{K+1}$ such that an option with index in $\mathcal{K}_1$ is selected with probability $1$ over the distribution of the possible valuations. Similarly, such a set $\mathcal{K}_2$ exists for $M_2$. To define the curve that provides mode connectivity, we need the following technical lemma, which is proven in \Cref{sec:rochetproofs}.

\begin{lemma}\label{lem:bijection}
    There exists a bijection $\varphi\colon \mathcal{K} \to \mathcal{K}_1\times \mathcal{K}_2$ such that for all $k\in \mathcal{K}_1$ we have that $\varphi(k)\in \{k\}\times \mathcal{K}_2$, and for all $k\in \mathcal{K}_2$ we have that $\varphi(k)\in \mathcal{K}_1\times \{k\}$.
\end{lemma}

With this lemma, we can define $\widehat M_1$ and $\widehat M_2$. Let $\varphi$ the bijection from \Cref{lem:bijection} and suppose $M_1=\{(x^{(k)}, p^{(k)})\}_{k \in \mathcal{K}}$. We then define $\widehat M_1=\{(x^{(\varphi_1(k))}, p^{(\varphi_1(k))})\}_{k\in\mathcal{K}}$, where $\varphi_1(k)$ is the first component of~$\varphi(k)$. Similarly, $\widehat M_2$ is derived from $M_2$ by using the second component $\varphi_2(k)$ of $\varphi(k)$ instead of $\varphi_1(k)$.
It remains to show that all menus on the three straight line segments from $ M_1$ via $\widehat M_1$ and $\widehat M_2$ to~$ M_2$ yield a revenue of at least $\min\{\rev{M_1},\rev{M_2}\}$, which is established by the following two propositions; their proofs can be found in  \Cref{sec:rochetproofs}.

\begin{proposition}\label{prop:copying}
    Let $M=\lambda  M_1 + (1-\lambda) \widehat M_1$ be a convex combination of the menus $M_1$ and $\widehat M_1$. Then $\rev{M} \geq \rev{ M_1}$. Similarly, every convex combination of the menus $ M_2$ and $\widehat M_2$ has revenue at least~$\rev{M_2}$.
\end{proposition}
The idea to prove \Cref{prop:copying} is that, on the whole line segment from $M_1$ to $\widehat M_1$, the only active options are those in $\mathcal{K}'$, implying that the revenue does not decrease. 

\begin{proposition}\label{prop:middle_piece}
    Let $M=\lambda \widehat M_1 + (1-\lambda) \widehat M_2$ be a convex combination of the menus $\widehat M_1$ and $\widehat M_2$. Then, $\rev{M} \geq \lambda \rev{\widehat M_1} + (1-\lambda) \rev{\widehat M_2}$.
\end{proposition}

The idea to prove \Cref{prop:middle_piece} is that, due to the special structure provided by \Cref{lem:bijection}, a linear interpolation between the menus also provides a linear interpolation between the revenues. Note that without the  construction of \Cref{lem:bijection}, such a linear relation would be  false; such an  example is shown in \Cref{sec:example-non-concavity}. 

    \Cref{prop:zero}  directly follows from \Cref{prop:copying} and \Cref{prop:middle_piece}.
Based on \Cref{prop:zero}, we can show our two main theorems for the RochetNet. The first result follows relatively easily from \Cref{prop:zero}.

\begin{theorem} \label{thm:epsilon-reducible}
    If two menus $M_1$ and $M_2$ for the RochetNet are $\varepsilon$-reducible, then they are $\varepsilon$-mode-connected. Moreover, the curve transforming $M_1$ into $M_2$ is piecewise linear with only five pieces.
\end{theorem}
\begin{proof}
    We prove this result by showing that every $\varepsilon$-reducible menu $M$ can be linearly transformed into a $0$-reducible menu $\widetilde M$ such that each convex combination of $M$ and $\widetilde M$ achieves a revenue of at least \mbox{$\rev{M}-\varepsilon$}. This transformation converting $M_1$ and $M_2$ to $\widetilde M_1$ and $\widetilde M_2$, respectively, yields the first and the fifth of the linear pieces transforming $M_1$ to $M_2$. Together with \Cref{prop:zero} applied to $\widetilde M_1$ and $\widetilde M_2$ serving as the second to fourth linear piece; the theorem then follows.

    To this end, let $M$ be an $\varepsilon$-reducible menu with options indexed by $k\in\mathcal{K}$. By definition, there is a subset $\mathcal{K}'\subseteq\mathcal{K}$ of at most $\sqrt{K+1}$ many options such that, with probability at least $1-\varepsilon$, the assigned active option is contained in $\mathcal{K}'$. Let $\widetilde M$ consist of the same allocations as $M$, but with modified prices. For an option $k\in\mathcal{K}'$, the price $\tilde p^{(k)} = p^{(k)}$ in $\widetilde M$ is the same as in $M$. However, for an option $k\in\mathcal{K}\setminus\mathcal{K}'$, we set the price $\tilde p^{(k)} >1$ in $\widetilde M$ to be larger than the largest possible valuation of any option $\lVert v \rVert_1\leq 1$. It follows that such an option will never be selected and $\widetilde M$ is $0$-reducible.

    To complete the proof, let us look at the reward of a convex combination $M'=\lambda M + (1-\lambda) \widetilde M$. If for a particular valuation $v$ the selected option in $M$ was in $\mathcal{K}'$, then the same option will be selected in $M'$. This happens with probability at least $1-\varepsilon$. In any other case, anything can happen, but the revenue cannot worsen by more than the maximum possible valuation, which is $\lVert v \rVert_1\leq 1$. Therefore, $\rev{M}-\rev{M'}\leq \varepsilon\cdot 1 = \varepsilon$, completing the proof.
\end{proof}

\begin{theorem}
    \label{thm:large}
    If two menus $M_1$ and $M_2$ for the RochetNet have size at least $\xxs$, then they are $\varepsilon$-connected. Moreover, the curve transforming $M_1$ into $M_2$ is piecewise linear with only five pieces.
\end{theorem}
\begin{proof}[Proof Sketch]
    The full proof can be found in \Cref{sec:large-disc}.
    The intuition behind this theorem is that if menus are large, then they should contain many redundant options. Indeed, as in the previous theorem, the strategy is as follows. We show that every menu $M$ of size at least $\xxs$ can be linearly transformed into a $0$-reducible menu $\widetilde M$ such that each convex combination of $M$ and $\widetilde M$ achieves a revenue of at least \mbox{$\rev{M}-\varepsilon$}. This transformation converting $M_1$ and $M_2$ to $\widetilde M_1$ and $\widetilde M_2$, respectively, yields the first and the fifth of the linear pieces transforming $M_1$ to $M_2$. Together with \Cref{prop:zero} applied to $\widetilde M_1$ and $\widetilde M_2$ serving as the second to fourth linear piece, the theorem then follows.

    However, this time, the linear transformation of $M$ to $\widetilde M$ is much more intricate than in the previous theorem. To do so, it is not sufficient to only adapt the prices. Instead, we also change the allocations of the menu options by rounding them to discretized values. This technique is inspired by \citet{dughmi2014sampling}, but non-trivially adapted to our setting. Since the rounding may also modify the active option for each valuation, we have to carefully adapt the prices in order to make sure that for each valuation, the newly selected option is not significantly worse than the originally selected one. Finally, this property has to be proven not only for $\widetilde M$, but for every convex combination of $M$ and $\widetilde M$.

    After the above rounding, the number of possible allocations for any option is bounded by $\xx$. Out of several options with the same allocation, the buyer would always choose the cheapest one, implying that the resulting menu $\widetilde M$ is $0$-reducible.
\end{proof}
\section{Mode Connectivity for the Affine Maximizer Auctions} \label{sec:AMA}
Throughout this section, we focus on AMAs with fixed weights $w_i=1$ for all buyers $i$.
Similarly to \emph{RochetNet}, we have the following definition for AMAs.

\begin{definition}\label{def:ama-reducible}
   A menu $M$ with $K+1$ options is $\varepsilon$-reducible if and only if there exists a subset $\mathcal{K}'\subseteq \mathcal{K}$, $0\in\mathcal{K}'$,  $|\mathcal{K}'|\le \sqrt{K+1}$ such that, with probability at least $1 - \frac{\varepsilon}{m}$ over the distribution of the valuation of the buyers, (i) $k(v_{-i}) \in \mathcal{K}'$ for any buyer $i$; and (ii) $k(v) \in \mathcal{K}'$.
\end{definition}
Such phenomena are observed in the experiments in \citep[Section~6.3]{curry2022differentiable}.

Our two main results, namely that two $\varepsilon$-reducible menus are always $\varepsilon$-connected and two large menus are always $\varepsilon$-connected, are based on the following proposition, in which we show that $0$-reducibility implies $0$-connectivity.

\begin{proposition} \label{prop:AMA-zero}
    If two menus $M_1$ and $M_2$ are $0$-reducible,  then they are $0$-connected. Moreover, the curve transforming $M_1$ into $M_2$ is piecewise linear with only three pieces. 
\end{proposition}
The  proof idea is similar to the proof of Proposition~\ref{prop:zero} in \emph{RochetNet}, {but requires additional arguments due to the more intricate price structure} (see Appendix~\ref{app:AMA-zero} for more details).
Based on this proposition, now, we are able to show our two main results. First, we achieve $\varepsilon$-connectivity from $\varepsilon$-reducibility. 
\begin{theorem} \label{thm:AMA-1}
    If two AMAs $M_1$ 
    and $M_2$ 
    are $\varepsilon$-reducible, then they are $\varepsilon$-mode-connected. Moreover, the curve transforming $M_1$ to $M_2$ is piecewise linear with only five pieces.
\end{theorem}
Before the proof, we recall how the total payment is calculated for a  valuation profile $v$.  We choose $k(v)$ as the option which maximizes the boosted welfare, $\sum_i v_i^\top x_{i}^{(k)} + \boost^{(k)}$.  
According to \eqref{eq:prices-formula}, the total revenue can be written as
\begin{equation}\label{eq:total-payment}
\begin{aligned}
\sum_i p_i(v)=
    \sum_i \underbrace{{\left(\sum_{\ell \neq i} v_{\ell}^\top x_{\ell}^{(k(v_{-i}))} + \boost^{(k(v_{-i}))}\right)}}_{\text{boosted welfare of $v_{-i}$}} - (m - 1) \underbrace{\left(\sum_i v_{i}^\top x_{i}^{(k(v))} + \boost^{(k(v))} \right)}_{\text{boosted welfare of $v$}} - \boost^{(k(v))}.
\end{aligned}    
\end{equation}

\begin{proof}[Proof of Theorem~\ref{thm:AMA-1}]
Similar to the proof of Theorem~\ref{thm:epsilon-reducible}, it is sufficient to show that every $\varepsilon$-reducible menu $M$ can be linearly transformed into a $0$-reducible menu $\widetilde M$ such that each convex combination of $M$ and $\widetilde M$ achieves a revenue of at least $\rev{M} - \varepsilon$. {This can then be used as the first and fifth linear piece of the curve connecting $M_1$ and $M_2$, while the middle three pieces are provided by \Cref{prop:AMA-zero}.}

We construct $\widetilde M$ by {\em (i)} keeping all options in~$\mathcal{K}'$ unchanged; {\em(ii)} for the options $k \in\options\setminus \mathcal{K}'$, we decrease $\boost^{(k)}$ to be smaller than $- m$, which implies such an option will never be selected (recall that $0\in \options'$ is assumed, and the option $(\mathbf{0},0)$ is better than any such option). Consequently, $\widetilde M$ is $0$-reducible.  

    To complete the proof, let us look at the revenue of $M' = \{ ({x'}^{(k)},{\boost'}^{(k)})\}_{k\in \options}$, which is a convex combination of $M$ and $\widetilde M$: $M' = \lambda M + (1 - \lambda) \widetilde M$ for $0 \leq \lambda < 1$. Let ${k'}(v) = \arg \max_k \sum_{i} v_{i}^\top {x'}^{(k)}_{i} + {\boost'}^{(k)}$. As we decrease $\boost^{(k)}$ for $k \notin \mathcal{K}'$,  $k(v) \in \mathcal{K}'$ implies ${k'}(v)  \in \mathcal{K}'$ and, additionally, option ${k'}(v)$ and option $k(v)$ achieve the same boosted welfare and same $\boost$. Therefore, since $M$ is $\varepsilon$-reducible, with probability at least $1 - \frac{\varepsilon}{m}$, the boosted welfare of $v$ as well as the boosted welfare of $v_{-i}$ for all buyers $i$ is the same for $M$ and for $M'$. According to the formula \eqref{eq:total-payment}, the total payment for the profile $v$ is the same for $M$ and $M'$.
Therefore, the  loss on the revenue can only appear with probability at most  $\frac{\varepsilon}{m}$,  and the maximum loss is at most $m$, which implies an $\varepsilon$ loss in total.
\end{proof}
Second, we show that mode connectivity also holds for those AMAs with large menu sizes, namely for $K+1 \geq \yys$.
\begin{theorem}\label{thm:AMA-2}
     For any $0 < \epsilon \leq \frac{1}{4}$, if two AMAs $M_1$ 
     and $M_2$ 
     have at least $K+1 \geq \yys$ options, then they are $\varepsilon$-mode-connected. Moreover, the curve transforming $M_1$ to $M_2$ is piecewise linear with only five pieces.
\end{theorem}
\begin{proof}[Proof Sketch]
The full proof can be found in Appendix~\ref{sec:AMA-discrete}.
Similar to \emph{RochetNet}, the idea of proving Theorem~\ref{thm:AMA-2} is to discretize the allocations in the menu, then one can use Proposition~\ref{prop:AMA-zero} to construct the low-loss transformation from $M_1$ to $M_2$ by five linear pieces. To do this, one wants the loss of revenue to be small during the discretization. Consider the formula \eqref{eq:total-payment} of the total payment.
The first two terms do not change much by a small change of the discretization. However, the last term~$\boost^{(k(v))}$ might be significantly affected by discretization, which may cause a notable decrease in the total payment. To avoid this, we perform a proportional discount on  $\boost$, incentivizing the auctioneer to choose an allocation with a small $\boost$. By this approach, the original revenue will be approximately maintained. Furthermore, we show a linear path, connecting the original menu and the menu after discretizing, which will suffer a small loss. 
\end{proof}


\section{Conclusion}
We have given theoretical evidence of mode-connectivity in neural networks designed to learn auction mechanisms. Our results show that, for a sufficiently wide hidden layer, $\varepsilon$-mode-connectivity holds in the strongest possible sense. Perhaps more practically, we have shown $\varepsilon$-mode-connectivity under $\varepsilon$\nobreakdash-reducibility, i.e., the assumption that there is a sufficiently small subset of neurons that preserve most of the revenue. There is evidence for this assumption in previous work in differentiable economics. A systematic experimental study that verifies this assumption under various distributions and network sizes is left for future work.

Our results make a first step in providing theoretical arguments underlying the success of neural networks in mechanism design. Our focus was on some of the most basic architectures. A natural next step is to extend the arguments for AMA networks with variable weights $w_i$. Such a result will need to analyze a four layer network, and thus could make headway into understanding the behaviour of deep networks. Besides \emph{RochetNet}, \citet{dutting2019optimal} also proposed \emph{RegretNet}, based on minimising a regret objective. This network is also applicable to multiple buyers, but only provides approximate incentive compatibility, and has been extended in subsequent work, e.g., \cite{feng2018deep,golowich2018deep,duan2022context}.  The architecture  is however quite different from \emph{RochetNet}: it involves two deep neural networks in conjunction, an allocation and a payment network, and uses expected ex post regret as the loss function.
{We therefore expect a mode-connectivity analysis for \emph{RegretNet} to require a considerable extension of the techniques used by us. We believe that such an analysis} would be a significant next step in the theoretical analysis of neural networks in differentiable economics.

\begin{ack}
    All three authors gratefully acknowledge support by the European Research Council (ERC) under the European Union’s Horizon 2020 research and innovation programme (for all three authors via grant agreement ScaleOpt–757481; for Christoph Hertrich additionally via grant agreement ForEFront--615640). Yixin Tao also acknowledges the Grant 2023110522 from SUFE.
\end{ack}

\bibliographystyle{abbrvnat}
\bibliography{sample}

\newpage
\appendix

\begin{center}
	\bfseries
	\Large
	Supplemental:\\Mode Connectivity in Auction Design
\end{center}

\section{Detailed Proofs of the Mode Connectivity for the RochetNet}
\label{sec:rochetproofs}

In this section we provide the detailed proofs omitted in \Cref{sec:rochetnet}.

\subsection{Interpolating between 0-reducible menus}

We start with the proofs of statements on the way towards proving \Cref{prop:zero}.

\begin{proof}[Proof of \Cref{lem:bijection}]
    We prove the claim by providing an explicit construction for $\varphi$ in two different cases.
    
    First, suppose that there is some $k^*\in \mathcal{K}_1\cap \mathcal{K}_2$. In this case, start by setting $\varphi(k)\coloneqq(k,k)$ for all~$k\in \mathcal{K}_1\cap \mathcal{K}_2$. Then, for all $k\in \mathcal{K}_1\setminus \mathcal{K}_2$, set $\varphi(k)=(k,k^*)$, and for all $k\in \mathcal{K}_2\setminus \mathcal{K}_1$, set $\varphi(k)=(k^*,k)$. So far, we have not assigned any pair twice and the two conditions of the lemma are already satisfied, so we can simply assign the remaining pairs arbitrarily.

    Second, suppose that $\mathcal{K}_1$ and $\mathcal{K}_2$ are disjoint. Note that for this being possible, $\sqrt{K+1}$ must be at least~$2$. Pick some distinct $k_1,k'_1\in \mathcal{K}_1$ and $k_2,k'_2\in \mathcal{K}_2$. Set $\varphi(k_1)\coloneqq(k_1,k_2)$, $\varphi(k'_1)\coloneqq(k'_1,k'_2)$, $\varphi(k_2)\coloneqq(k'_1,k_2)$, and $\varphi(k'_2)\coloneqq(k_1,k'_2)$. Then, for all $k\in \mathcal{K}_1\setminus \{k_1,k'_1\}$, set $\varphi(k)\coloneqq (k,k_2)$ and for all $k\in \mathcal{K}_2\setminus \{k_2,k'_2\}$, set $\varphi(k)\coloneqq (k_1,k)$. Again, we have not assigned any pair twice and the two conditions of the lemma are already satisfied, so we can simply assign the remaining pairs arbitrarily.
\end{proof}

\begin{proof}[Proof of \Cref{prop:copying}]
    We only prove the first statement on $M_1$; the statement on $M_2$ follows analogously.
 We show that for each possible valuation $v$ of the buyer, the price paid to the seller for menu $M$ is at least as high as in menu $M_1$. Suppose for valuation $v$ that the buyer chooses the $k$-th option in menu $M_1$. Note that we may assume $k\in \mathcal{K}_1$ due to $0$-reducibility of $M_1$. By construction of~$\varphi$, it follows that $\varphi_1(k)=k$. Therefore, the $k$-th option in~$M$ is exactly equal to the $k$-th option in $M_1$. Making use of the fact that ties are broken in favor of larger prices, it suffices to show that the $k$-th option is utility-maximizing in $M$, too.

    To this end, let $k'\in\mathcal{K}$ be an arbitrary index. If $M_1=\{({x}^{(k)},p^{(k)}\}_{k\in\mathcal{K}}$, then the utility of option $k'$ in $M$ is
    \begin{align*}
        &\phantom{{}={}} v^\top(\lambda {x}^{(k')} + (1-\lambda) {x}^{(\varphi_1(k'))}) - (\lambda{p}^{(k')} + (1-\lambda) {p}^{(\varphi_1(k'))})\\
        &= \lambda (v^\top {x}^{(k')}- {p}^{(k')}) + (1-\lambda) (v^\top {x}^{(\varphi_1(k'))}- {p}^{(\varphi_1(k'))})\\
        &\leq \lambda (v^\top {x}^{(k)}- {p}^{(k)}) + (1-\lambda) (v^\top {x}^{(k)}- {p}^{(k)})\\
        &= v^\top {x}^{(k)}- {p}^{(k)},
    \end{align*}
    where the inequality follows because the $k$-th option is utility-maximizing for menu $M_1$. This shows that it is utility-maximizing for menu $M$, completing the proof.
\end{proof}

\begin{proof}[Proof of \Cref{prop:middle_piece}]
    The claim is trivial for $\lambda=0$ or $\lambda=1$. Therefore, assume $0<\lambda<1$ for the remainder of the proof.
    Again we show that the claim holds pointwise for each possible valuation and therefore also for the revenue. For valuation $v$, let $k_1$ and $k_2$ be the active option assigned to the buyer in $\widehat M_1$ and $\widehat M_2$, respectively. Note that by construction of the menus $\widehat M_1$ and $\widehat M_2$ we may assume without loss of generality that $k_1\in \mathcal{K}_1$ and $k_2\in \mathcal{K}_2$. Let $k^*\coloneqq \varphi^{-1}(k_1, k_2)$.
    
    We show that option $k^*$ is utility-maximizing in $M=\{(x^{(k)},p^{(k)})\}_{k\in\mathcal{K}}$. To this end, we use the notation $\widehat M_1 = \{(\hat{x}^{(k)}, \hat{p}^{(k)})\}_{k\in\mathcal{K}}$ and $\widehat M_2 = \{(\hat{y}^{(k)}, \hat{q}^{(k)})\}_{k\in\mathcal{K}}$. Let $k'\in\mathcal{K}$ be an arbitrary index. The utility of option $k'$ in menu $M$ can be bounded as follows:
    \begin{align*}
        v^\top x^{(k')} - p^{(k')}
        &= \lambda (v^\top \hat{x}^{(k')} - \hat{p}^{(k')}) + (1-\lambda) (v^\top \hat{y}^{(k')} - \hat{q}^{(k')})\\
        &\leq \lambda (v^\top \hat{x}^{(k_1)} - \hat{p}^{(k_1)}) + (1-\lambda) (v^\top \hat{y}^{(k_2)} - \hat{q}^{(k_2)})\\
        &= \lambda (v^\top \hat{x}^{(k^*)} - \hat{p}^{(k^*)}) + (1-\lambda) (v^\top \hat{y}^{(k^*)} - \hat{q}^{(k^*)})\\
        &= v^\top x^{(k^*)} - p^{(k^*)},
    \end{align*}
    where the inequality in the second line follows because $k_1$ and $k_2$ are utility-maximizing for $\widehat M_1$ and~$\widehat M_2$, respectively, and the equality in the third line follows because, by construction, in menu~$\widehat M_1$ option~$k^*$ is equivalent to option $k_1=\varphi_1(k^*)\in \mathcal{K}_1$, and similarly in menu $\widehat M_2$ option $k^*$ is equivalent to option $k_2=\varphi_2(k^*)\in \mathcal{K}_2$. This concludes the proof that $k^*$ is utility-maximizing.

    With the same reasoning as above, we obtain $p^{(k^*)}= \lambda \hat{p}^{(k_1)} + (1-\lambda) \hat{q}^{(k_2)}$, from which we conclude that the price  achieved by the seller in menu $M$ for valuation $v$ is at least as high as the convex combination of the achieved prices for menus $\widehat M_1$ and $\widehat M_2$.
\end{proof}

\subsection{Discretizing large menus}\label{sec:large-disc}

This subsection is devoted to providing a detailed proof of \Cref{thm:large}. To do so, we will show how to convert any menu $M$ of size at least $\xxs$ into a $0$-reducible menu $\widetilde M$ such that each convex combination of $M$ and $\widetilde M$ achieves a revenue of at least $\rev{M}-\varepsilon$. Without loss of generality, we assume that $M$ has size exactly $K+1=\xxs$.

To construct the menu $\widetilde{M}$ satisfying these requirements, we adapt techniques from \cite{dughmi2014sampling}.\footnote{In their paper, they construct a menu with a finite number of options to approximate the optimal mechanism. The approximation is based on the multiplicative error, and they assume the buyer's valuation is no less than $1$.}
In general, the idea is to discretize the allocations in the menu by a finite allocation set~$S$ (see Definition~\ref{def:dis-alloc}) whose size is at most $\sqrt{K+1} = \xx$. However, because of the discretization, the buyer may choose an option with a much smaller price, providing a lower revenue compared to the original menu. To deal with this, we also decrease the prices on the menu; the decrease is in proportion to the price. Intuitively, this incentives the buyer to choose the option with an originally high price. We show, after this modification, the menu achieves a revenue of at least $\rev{M} - \varepsilon$.

For ease of notation, we will use $\tilde{\varepsilon}\coloneqq\frac{\varepsilon^2}{4}$ and, therefore, $2\sqrt{ \tilde{\varepsilon}} = \varepsilon$.

\begin{definition} \label{def:dis-alloc}
    Let $S$ be a (finite) set of allocations.
    We say that $S$ is an $\tilde{\varepsilon}$-cover if, for every possible allocation $x$, there exists an allocation $\tilde{x} \in S$ such that for every possible valuation vector $v$ we have that $v^\top x \geq v^\top \tilde{x} \geq v^\top x - \tilde{\varepsilon}$. 
\end{definition}

The following proposition shows that one can construct an $\tilde{\varepsilon}$-cover $S$ with size at most $\xx$.
\begin{proposition}
If $\lVert v \rVert_1\leq 1$, then \[S = \underbrace{\left\{\tilde{\varepsilon} s\right\}_{s = 0}^{\left\lfloor \frac{1}{\tilde{\varepsilon}} \right \rfloor} \times \left\{\tilde{\varepsilon} s\right\}_{s = 0}^{\left\lfloor \frac{1}{\tilde{\varepsilon}} \right \rfloor} \times \cdots \times \left\{\tilde{\varepsilon} s\right\}_{s = 0}^{\left\lfloor \frac{1}{\tilde{\varepsilon}} \right \rfloor}}_\text{$n$ terms} \] is an $\tilde{\varepsilon}$-cover, and $|S| = \left \lceil\frac{1}{\tilde{\varepsilon}} \right\rceil^n = \xx$.
\end{proposition}
\begin{proof}
    For any allocation $x$, we can round it down to $\tilde{x}$, such that  $\tilde{x}_j =  \lfloor \frac{x_j}{\tilde{\varepsilon}} \rfloor \cdot \tilde{\varepsilon}$. It is not hard to see that $v^\top x \geq v^\top \tilde{x}$. Additionally, the inequality $v^\top \tilde{x} \geq v^\top x - \tilde{\varepsilon}$ follows as the total loss is at most $v^\top (\tilde{x} -  x) \leq \|v\|_1 \|\tilde{x} -  x\|_{\infty} \leq \tilde{\varepsilon}$.
\end{proof}

\paragraph{Construction of $\widetilde{M}$.} Given $S$, we can construct $\widetilde{M}$ as follows.
Each option $(x^{(k)}, p^{(k)})$ in menu $M$ is modified to $(\tilde{x}^{(k)}, \tilde{p}^{(k)})$ in menu $\widetilde{M}$, where $\tilde{x}^{(k)}$ is the corresponding allocation of $x^{(k)}$ in $S$ and the price is set to $\tilde{p}^{(k)} = \left(1 - \sqrt{\tilde{\varepsilon}}\right) p^{(k)}$:
\[
    \widetilde{M} = \left\{\left(\tilde{x}^{(k)}, \tilde{p}^{(k)} \right)\right\}_{k\in\mathcal{K}}.
\]

The following lemma shows that this construction indeed ensures that the reward decreases by at most~$\varepsilon$.

\begin{lemma} \label{lem:const-menu-size}
It holds that $\rev{\widetilde{M}} \geq \rev{M} - 2 \sqrt{ \tilde{\varepsilon}} = \rev{M} - \varepsilon$.
\end{lemma}
\begin{proof}
The following inequalities demonstrate the buyer who chooses option $k$ in menu $M$ will not choose option $k'$ in menu $\widetilde{M}$ such that $p^{(k')} < p^{(k)} - \sqrt{ \tilde{\varepsilon}}$.
\begin{align*}
v^\top \tilde{x}^{(k)} - \left(1 - \sqrt{\tilde{\varepsilon}}\right) p^{(k)} &\geq v^\top x^{(k)} -  p^{(k)} - \tilde{\varepsilon} + \sqrt{\tilde{\varepsilon}} p^{(k)}  \\
&\geq v^\top x^{(k')} -  p^{(k')} - \tilde{\varepsilon} + \sqrt{\tilde{\varepsilon}} p^{(k)} \\
&\geq v^\top \tilde{x}^{(k')} -  \left(1 - \sqrt{ \tilde{\varepsilon}}\right) p^{(k')} - \tilde{\varepsilon} + \sqrt{ \tilde{\varepsilon}} (p^{(k)} - p^{(k')}) \\
&>  v^\top \tilde{x}^{(k')} -  \left(1 - \sqrt{ \tilde{\varepsilon}}\right) p^{(k')}. \numberthis \label{eq:price-discount}
\end{align*}
The first and third inequalities hold by Definition~\ref{def:dis-alloc} and  the second inequality holds as the buyer will choose option $k$ in menu $M_1$.
    
Therefore, the total loss on the revenue is upper bounded by $\sqrt{ \tilde{\varepsilon}} p^{(k)} + \sqrt{\tilde{\varepsilon}} \leq 2 \sqrt{ \tilde{\varepsilon}}$, as the price satisfies $p^{(k)} \leq 1$.
\end{proof}

In addition to this property of $\widetilde{M}$ itself, we also need to show the revenue does not drop more than $\varepsilon$ for any menu on the line segment connecting $M$ to $\widetilde{M}$.

\begin{lemma}\label{lem:path_to_discretized}
Let $M'= \lambda M + (1-\lambda) \widetilde M$ be a convex combination of the menus $M$ and $\widetilde M$. Then, $\rev{M'} \geq \rev{M} - 2 \sqrt{  \tilde{\varepsilon}}  = \rev{M} - \varepsilon$.
\end{lemma}
\begin{proof}
Let $M=\{(x^{(k)}, p^{(k)})\}_{k\in\mathcal{K}}$ and $M'=\{(x'^{(k)}, p'^{(k)})\}_{k\in\mathcal{K}}$.
Similar to the proof of Lemma~\ref{lem:const-menu-size}, we show that the buyer who chooses option $k$ in menu $M$ will not choose option $k'$ in menu $M'$ such that $p^{(k')} < p^{(k)} - \sqrt{\tilde{\varepsilon}}$. This is true by the following (in)equalities. For any $k'\in\mathcal{K}$, we have that
 \begin{align*}
     v^\top {x'}^{(k)} - {p'}^{(k)} &= \lambda (v^\top x^{(k)} - p^{(k)}) + (1 - \lambda)(v^\top \tilde{x}^{(k)} - \tilde{p}^{(k)}) \\
     &> \lambda (v^\top x^{(k')} - p^{(k')}) + (1 - \lambda)(v^\top \tilde{x}^{(k')} - \tilde{p}^{(k')}).
 \end{align*}
 The inequality follows by combining (i) $v^\top x^{(k)} - p^{(k)} \geq v^\top x^{(k')} - p^{(k')}$, which is true as the buyer will choose option $k$ in menu $M$; and (ii) $v^\top \tilde{x}^{(k)} - \tilde{p}^{(k)} > v^\top \tilde{x}^{(k')} - \tilde{p}^{(k')} $ from \eqref{eq:price-discount}.

 Similar to the proof of Lemma~\ref{lem:const-menu-size}, it follows than that the total loss on the revenue is upper bounded by $2 \sqrt{ \tilde{\varepsilon}}$.
\end{proof}

With these lemmas at hand, we can finally prove \Cref{thm:large}.

\begin{proof}[Proof of \Cref{thm:large}]
    Applying the transformation described in this section to convert $M_1$ and $M_2$ results in two menus $\widetilde M_1$ and $\widetilde M_2$, respectively. Since $\widetilde M_1$ and $\widetilde M_2$ contain at most $\sqrt{K+1}=\xx$ different allocations and a buyer would always choose the cheapest out of several options with the same allocation, they are $0$-reducible. Applying \Cref{prop:zero} to them implies that they are $0$-mode-connected with three linear pieces. Combining these observations with \Cref{lem:const-menu-size,lem:path_to_discretized} implies that $M_1$ and $M_2$ are $\varepsilon$-connected with five linear pieces.
\end{proof}
\section{Bounds on the Error of the Softmax Approximation for the Argmax}\label{sec:softmax-rochet}
In the RochetNet, to ensure that the objective is a smooth function, a softmax operation is used instead of the argmax during the training process:  
$$\revsoftmax{M} = \int \sum_{k = 1}^K p_i \frac{e^{Y ({x^{(k)}}^\top v - p^{(k)})}}{\sum_{k' = 1}^K e^{Y ({x^{(k')}}^\top v - p^{(k')})}} \text{d} F(v).$$
Here, $Y$ is a sufficiently large constant. In this section, we will look at the difference between the actual revenue and this softmax revenue. 

We would like to assume the density of the valuation distribution is upper bounded by $\mathcal{X} = \max_{v\in [0,1]^n\text{ and }\|v\|_1 \leq 1} f(v)$, which is a finite value. 
Given this assumption, the following lemma shows that,  for any menu $M$ of size $K$, the difference between the actual revenue and the softmax revenue is bounded. 

\begin{lemma}\label{lem:softmax}
For any $M$ and $Y \geq 1$,
\begin{align*}
    |\revsoftmax{M} - \rev{M}| \leq \frac{K + 1}{Y} \left((n \mathcal{X} +1 + \frac{\mathcal{X}}{Y}) \log \frac{Y}{ \mathcal{X} } + \mathcal{X}\right).
\end{align*}
\end{lemma}
\begin{proof}
We prove $\revsoftmax{M} - \rev{M} \leq  \frac{K}{Y} \left((n \mathcal{X} +1 + \frac{\mathcal{X}}{Y}) \log \frac{Y}{ \mathcal{X} } + \mathcal{X}\right)$. $\rev{M} - \revsoftmax{M} \leq  \frac{K}{Y} \left((n \mathcal{X} +1 + \frac{\mathcal{X}}{Y}) \log \frac{Y}{ \mathcal{X} } + \mathcal{X}\right)$ follows by a similar argument.

 Let $k(v)$ be the option chosen in menu $M$ when the buyer's valuation is  $v$.  Then, the difference between these two can be bounded as follows. 
\begin{align*}
    \revsoftmax{M} - \rev{M}& \leq \int \sum_{k = 0}^K (p^{(k)} - p^{(k(v))})^+ \cdot \frac{e^{Y (v^\top {x^{(k)}} - p^{(k)})}}{\sum_{k' = 1}^K e^{Y (v^\top {x^{(k')}} - p^{(k')})}}  \text{d} F(v) \\
    &\leq \int \sum_{k = 0}^K (p^{(k)} - p^{(k(v))})^+ e^{Y (v^\top {x^{(k)}} - p^{(k)} - v^\top x^{(k(v))} + p^{(k(v))})}  \text{d} F(v).
\end{align*}
Here, $(\cdot)^+ \triangleq \max\{ \cdot, 0\}$. Now, we focus on one option $k$, and we will give an upper bound on 
\begin{align*}
    \int (p^{(k)} - p^{(k(v))})^+ \mathsf{1}_{v^\top {x^{(k)}} - p^{(k)} + \sigma \geq v^\top {x^{(k(v))}} - p^{(k(v))} \geq v^\top {x^{(k)}} - p^{(k)}  } \text{d} F(v) \numberthis \label{eq:upper-diff-rev}
\end{align*}
for the non-negative parameter $\sigma$, which will be specified later. Note that, it is always true that $v^\top {x^{(k(v))}} - p^{(k(v))} \geq v^\top {x^{(k)}} - p^{(k)}$. If $v^\top {x^{(k)}} - p^{(k)} + \sigma \geq v^\top {x^{(k(v))}} - p^{(k(v))}$ is not satisfied then $e^{Y (v^\top {x^{(k)}} - p^{(k)} - v^\top x^{(k(v))} + p^{(k(v))})} \leq e^{- Y \sigma}$. Therefore, if \eqref{eq:upper-diff-rev} is upper bounded by $\mathcal{C}(\sigma)$, then $\texttt{Rev}^{\texttt{softmax}}_M - \rev{M} \leq (K + 1) (\mathcal{C}(\sigma) + (1 + \sigma)e^{- Y \sigma})$. \footnote{Note that if $p^{(k)}\geq 1 + \sigma$ then $v^\top {x^{(k)}} - p^{(k)} + (p^{(k)} - 1) \leq  v^\top {x^{(k(v))}} - p^{(k(v))}$ as $\texttt{LHS} \leq 0$ and $\texttt{RHS} \geq 0$. Therefore, if $v^\top {x^{(k)}} - p^{(k)} + \sigma \geq v^\top {x^{(k(v))}} - p^{(k(v))}$ is not satisfied, then $(p^{(k)} - p^{(k(v))})^+ e^{Y (v^\top {x^{(k)}} - p^{(k)} - v^\top x^{(k(v))} + p^{(k(v))})}   \leq \max_{\sigma' \geq \sigma}\{(1 + \sigma') e^{- Y \sigma'}\}$. Note that $\max_{\sigma' \geq \sigma}\{(1 + \sigma') e^{- Y \sigma'}\} \leq (1 + \sigma) e^{- Y \sigma}$ when $ Y \geq 1$. }

Note that 
\begin{align*}
    &\int (p^{(k)} - p^{(k(v))})^+ \mathsf{1}_{v^\top {x^{(k)}} - p^{(k)} + \sigma \geq v^\top {x^{(k(v))}} - p^{(k(v))} \geq v^\top {x^{(k)}} - p^{(k)}  } \text{d} F(v) \\
    &~~~~~~\leq \sigma + \int \sum_{j = 1}^n v_j (x^{(k)}_j - x^{(k(v))}_j)^+ \mathsf{1}_{v^\top {x^{(k)}} - p^{(k)} + \sigma \geq v^\top {x^{(k(v))}} - p^{(k(v))} \geq v^\top {x^{(k)}} - p^{(k)}  } \text{d} F(v) .  
\end{align*}
The inequality follows as we consider the region of $v$ such that $v^\top {x^{(k)}} - p^{(k)} + \sigma \geq v^\top {x^{(k(v))}} - p^{(k(v))}$. Additionally, since $v_j \in [0, 1]$,

\begin{align*}
    &\int v_j (x^{(k)}_j - x^{(k(v))}_j)^+ \mathsf{1}_{v^\top {x^{(k)}} - p^{(k)} + \sigma \geq v^\top {x^{(k(v))}} - p^{(k(v))} \geq v^\top {x^{(k)}} - p^{(k)}  } \text{d} F(v)  \\
    &~~~~~~\leq \int (x^{(k)}_j - x^{(k(v))}_j)^+ \mathsf{1}_{v^\top {x^{(k)}} - p^{(k)} + \sigma \geq v^\top {x^{(k(v))}} - p^{(k(v))} \geq v^\top {x^{(k)}} - p^{(k)}  } \text{d} F(v). \\
\end{align*}
Now we fix all coordinates of valuation $v$ other than coordinate $j$. Note that, the function $v^\top {x^{(k(v))}} - p^{(k(v))} - v^\top {x^{(k)}} - p^{(k)}$ is a convex function on $v_j$ and $x^{(k)}_j - x^{(k(v))}_j$ is the negative gradient of this convex function. Since we are looking at the region such that the function  $v^\top {x^{(k(v))}} - p^{(k(v))} - v^\top {x^{(k)}} - p^{(k)}$ is bounded in $[0, \sigma]$, this direct imply
\begin{align*}
   \int_{v_j \in [0, 1]} (x^{(k)}_j - x^{(k(v))}_j)^+ \mathsf{1}_{v^\top {x^{(k)}} - p^{(k)} + \sigma \geq v^\top {x^{(k(v))}} - p^{(k(v))} \geq v^\top {x^{(k)}} - p^{(k)}  } \text{d} F(v) \leq \mathcal{X}\sigma.
\end{align*}
This implies $\texttt{Rev}^{\texttt{softmax}}_M - \rev{M}\leq (K + 1)( \sigma + n \mathcal{X} \sigma + (1 +\sigma)e^{- Y \sigma})$ which is upper bounded by $\frac{K + 1}{Y} \left((n \mathcal{X} +1 + \frac{\mathcal{X}}{Y}) \log \frac{Y}{ \mathcal{X} } + \mathcal{X}\right)$ by setting $\sigma = \frac{1}{Y} \log \frac{Y}{ \mathcal{X} }$.
\end{proof}
\section{Example: Disconnected Local Maxima} \label{sec:example-non-concavity}

This section shows that the revenue is not quasiconcave on $M$, and in fact it might have disconnected local maxima.  Recall that a function $g$ is quasiconcave if and only if, for any $x$, $y$ and $\lambda \in [0, 1]$,
\begin{align*}
    g(\lambda x + (1 - \lambda) y) \geq \min \{g(x), g(y)\}
\end{align*}
Hence, quasiconcavity implies $0$-mode-connectivity with a single straight-line segment.

We consider the case that there is only one buyer, one item, and one regular option on the menu. Consider the following value distribution $f$:
\begin{align*}
    f(x) = \begin{cases}
1.5  & 0 < x \leq \frac{1}{3} + 0.15\\
0 &  \frac{1}{3} + 0.15 < x \leq \frac{2}{3} + 0.15\\
1.5 &  \frac{2}{3} + 0.15 < x \leq 1. \numberthis \label{dist:quasi-concave}
\end{cases}
\end{align*}
With this probability distribution, we show the following result. As Figure~\ref{fig:my_label} shows, there are two local maxima so that any continuous curve connecting them has lower revenue than either endpoint. Hence, mode connectivity fails between these two points. We only give a formal proof of the fact that the revenue is not quasiconcave.
\begin{figure}[hbt]
\centering
    \includegraphics[width=0.5\textwidth]{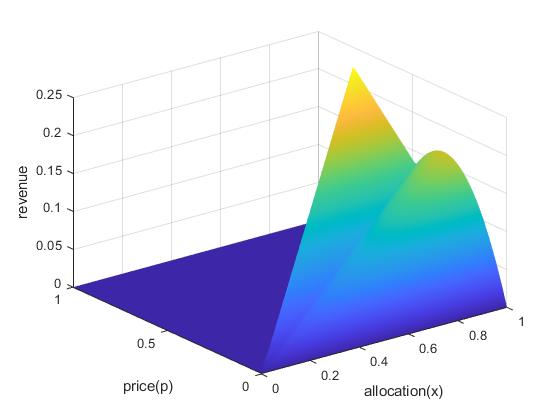}
    \caption{Revenue of the mechanism $M = \{(x, p)\}$ when the value distribution is $f$.}
    \label{fig:my_label}
\end{figure}
\begin{lemma}
$\rev{M}$ is not quasiconcave on $M$.
\end{lemma}
\begin{proof}
We consider the case where $n = 1$ (single item case); $K = 1$ (menu with single options). The value distribution, $f$ is defined in \eqref{dist:quasi-concave}. 

We consider two menus: $M_1$ and $M_2$, where $M_1 = \{(0,0),(1, 0.36)\}$ and $M_2 = \{(0,0),(1, 0.84)\}$. Then, $\rev{M_1} = 0.1656$ and $\rev{M_2} = 0.2016$.

However, if we consider $M_3 = \frac{1}{2} (M_1 + M_2) = \{(0,0), (1, 0.6)\}$, then this provides a revenue of $0.165$, which is strictly smaller than $\rev{M_1}$ and $\rev{M_2}$. More intuitively, Figure~\ref{fig:my_label} shows the revenue for $x \in [0, 1]$ and $p \in [0, 1]$.
\end{proof}

\section{Detailed Proofs of the Mode Connectivity for AMAs}

In this section, we provide the detailed proofs omitted in Section~\ref{sec:AMA}

\subsection{Interpolating between 0-reducible menus} \label{app:AMA-zero}

In this subsection, we will prove Proposition~\ref{prop:AMA-zero}, that is, we show that two $0$-reducible menus $M_1 = \{(x^{ (1, k)}, \boost^{(1, k)})\}_{k\in\mathcal{K}}$ and $M_2 = \{(x^{(2, k)}, \boost^{ (2, k)})\}_{k\in\mathcal{K}}$ are $0$-mode-connected.

Similar to \emph{RochetNet}, we introduce two intermediate menus $\widehat M_1$ and $\widehat M_2$, and show that every menu in the piecewise linear interpolation form $M_1$ via $\widehat M_1$ and $\widehat M_2$ to $M_2$ yields a revenue of at least $\min \{\rev{M_1}, \rev{M_2}\}$. Using that menu $M_1$ has only $\sqrt{K+1}$ non-redundant options, menu $\widehat M_1$ will be defined by repeating each of the $\sqrt{K+1}$ options $\sqrt{K+1}$ times. Menu $\widehat M_2$ will be derived from $M_2$ similarly.

To make this more formal, let $\mathcal{K}'_1$( and $\mathcal{K}'_2$) denote the set of the indexes of options in $M_1$( and $M_2$) in definition of $\varepsilon$-reducibility, respectively. 
Similar to \emph{RochetNet}, with the help of the Lemma~\ref{lem:bijection}, we can formally define $\widehat M_1$ and $\widehat M_2$ as $\widehat M_1=\{(x^{ (1, \varphi_1(k))}, \boost^{(1, \varphi_1(k))})\}_{k\in\mathcal{K}}$, where $\varphi_1(k)$ is the first component of~$\varphi(k)$; and, similarly, $\widehat M_2$ is derived from $M_2$ by using the second component $\varphi_2(k)$ of $\varphi(k)$ instead of $\varphi_1(k)$.

It remains to show that all menus on the three straight line segments from $M_1$ via $\widehat M_1$ and $\widehat M_2$ to~$M_2$ yield revenue of at least $\min\{\rev{ M_1},\rev{ M_2}\}$.

\begin{proposition}
    Let $M=\lambda  M_1 + (1-\lambda) \widehat M_1$ be a convex combination of the menus $ M_1$ and $\widehat M_1$. Then $\rev{M} \geq \rev{M_1}$. Similarly, every convex combination of the menus $M_2$ and $\widehat M_2$ has revenue at least~$\rev{M_2}$.
\end{proposition}
\begin{proof}
    We only prove the first statement because the second one is analogous. We show that for each possible valuation $v\in V^{mn}$ (with $\|v_i\| \leq 1$ for all $i$) of the buyers, the total payment paid to the auctioneer for menu $M$ is at least as high as in menu $M_1$. Suppose for valuation $v\in V^n$ that the auctioneer chooses the $k(v)$-th option in menu $M_1$ in maximizing the  boosted welfare. Note that we may assume  $k(v) \in \mathcal{K}'_1$ due to $0$-reducibility of $M_1$. By construction of~$\varphi$, it follows that $\varphi_1(k(v))=k(v)$. Therefore, the $k(v)$-th option in $M$ exactly equals the $k(v)$-th option in $M_1$. Because ties are broken in favor of larger total payments, it suffices to show that the $k(v)$-th option is the one with the highest  boosted welfare also in $M$. \footnote{Recall that the auctioneer will choose the option $k$ maximize the $\boost^{(k)}$ among all boosted welfare maximizing options given the formula of the total payment \eqref{eq:total-payment}.}
    
    Let $k'\in \mathcal{K}$ be an arbitrary index. The boosted welfare of option $k'$ in $M$ is 
    \begin{align*}
        \sum_{i}   v_{i}^\top (\lambda x_{i}^{(1,k')} &+ (1-\lambda) x_{i}^{(1,\varphi_1(k'))})) + (\boost^{(1, k')} + (1-\lambda) \boost^{(1, \varphi_1(k'))}) \\
        &= \lambda (\sum_{i}   v_{i}^\top x_{i}^{(1, k')} +  \boost^{(1, k')}) + (1-\lambda) (\sum_{i}   v_{i}^\top x_{i}^{(1, \varphi_1(k'))}+ \boost^{(1, \varphi_1(k'))})\\
        &\leq \lambda (\sum_{i}   v_{i}^\top x_{i}^{(1, k(v))} +  \boost^{(1, k(v))}) + (1-\lambda) (\sum_{i}   v_{i}^\top x_{i}^{(1, k(v))}+ \boost^{(1, k(v))})\\
        &= \sum_{i}   v_{i}^\top x^{(1, k(v))}_i +  \boost^{(1, k(v))},
    \end{align*}
    where the inequality follows because the $k(v)$-th option is  boosted welfare maximizing for menu $M_1$. This shows that $k(v)$ is also a boosted welfare maximizer for menu $M$, completing the proof.
\end{proof}

\begin{proposition}
    Let $M=\lambda \widehat M_1 + (1-\lambda) \widehat M_2$ be a convex combination of the menus $\widehat M_1$ and $\widehat M_2$. Then  $\rev{M} \geq \lambda \rev{\widehat M_1} + (1-\lambda) \rev{\widehat M_2}$.
\end{proposition}
\begin{proof}
    The claim is trivial for $\lambda=0$ or $\lambda=1$. Therefore, assume $0<\lambda<1$ for the remainder of the proof.
    For possible valuation $v\in V^n$ such that $\|v_i\| \leq 1$ for all $i$, let $k_1(v)$ and $k_2(v)$ be the boosted welfare maximizing options  in $\widehat M_1$ and $\widehat M_2$, respectively. Note that by the construction of the menus $\widehat M_1$ and $\widehat M_2$, we may assume without loss of generality that $k_1(v)\in \mathcal{K}'_1$ and $k_2(v)\in \mathcal{K}'_2$. Let $k^*(v)\coloneqq \varphi^{-1}(k_1(v), k_2(v))$. 
    
    We show that option $k^*(v)$ is  boosted welfare maximizing in $M=\{(x^{{(k)}},\boost^{{(k)}})\}_{k\in\mathcal{K}}$ with valuation $v$. To this end, we use the notation $\widehat M_1 = \{(\hat{x}^{{(1, k)}}, \hat{\boost}^{(1, k)})\}_{k\in\mathcal{K}}$ and $\widehat M_2 = \{(\hat{x}^{{(2, k)}}, \hat{\boost}^{(2, k)})\}_{k\in\mathcal{K}}$. Let $k'\in \mathcal{K}$ be an arbitrary index. Then, the boosted welfare of option $k'$ can be bounded as follows:
    \begin{align*}
        \sum_{i}   v_{i}^\top x_{i}^{(k')} + \boost^{(k')}&= \lambda (\sum_{i}   v_{i}^\top \hat{x}_{i}^{(1, k')} + \hat{\boost}^{(1, k')}) + (1-\lambda) (\sum_{i}   v_{i}^\top \hat{x}_{i}^{(2, k')} + \hat{\boost}^{(2, k')})\\
        &\leq \lambda (\sum_{i}   v_{i}^\top \hat{x}_{i}^{(1, k_1(v))} + \hat{\boost}^{(1, k_1(v))}) + (1-\lambda) (\sum_{ij} v_{i} \hat{x}_{ij}^{ (2, k_2(v))} + \hat{\boost}^{(2, k_2(v))}\\
        &= \lambda (\sum_{i}   v_{i}^\top \hat{x}_{i}^{(1, k^*(v))} + \hat{\boost}^{(1, k^*(v))}) + (1-\lambda) (\sum_{i}   v_{i}^\top \hat{x}_{i}^{(1, k^*(v))} + \hat{\boost}^{(1, k^*(v))})\\
        &=\sum_{i}   v_{i}^\top x_{i}^{(k^*(v))} + \boost^{(k^*(v))},
    \end{align*}
    where the inequality in the second line follows because $k_1(v)$ and $k_2(v)$ are boosted welfare maximizers for $\widehat M_1$ and~$\widehat M_2$, respectively. The equality in the third line follows because, by construction, in menu $\widehat M_1$ option~$k^*(v)$ is equivalent to option $k_1(v)=\varphi_1(k^*(v))\in \mathcal{K}'_1$. Similarly, in menu $\widehat M_2$ option $k^*(v)$ is equivalent to option $k_2(v)=\varphi_2(k^*(v))\in \mathcal{K}'_2$. This concludes the proof that $k^*(v)$ is a boosted welfare maximizer in $M$. 

    Note that the total payment of $M= \{(x^{{(k)}},\boost^{{(k)}})\}_{k\in\mathcal{K}}$ can be written in the following form:
    \begin{align*}
    \sum_i p_i(v) &= \sum_i    \left( \sum_{l \neq i}    v_{l}^\top x_{l}^{(k(v_{-i}))} + \boost^{(k(v_{-i}))} \right) - \sum_i   \left( \sum_{l \neq i}  v_{l}^\top x_{l}^{(k(v))} + \boost^{(k(v))} \right).
    \end{align*}
    where $k(\cdot)$ is the boosted welfare maximizer used in $M$.
    As  ties are broken in favor of larger total payments, this value decreases by replacing $k(\cdot)$ by $k^*(\cdot)$ \footnote{Note that both $k(\cdot)$ and $k^*(\cdot)$  maximize the boosted welfare.}: 
    \begin{align*}
    &\sum_i    \left( \sum_{l \neq i}   v_{l}^\top x_{l}^{(k(v_{-i}))} + \boost^{(k(v_{-i}))} \right)  - \sum_i   \left( \sum_{l \neq i} v_{l}^\top x_{l}^{(k(v))} + \boost^{(k(v))} \right) \\
    &\geq \sum_i    \left( \sum_{l \neq i} v_{l}^\top x_{l}^{(k^*(v_{-i}))} + \boost^{(k^*(v_{-i}))} \right)  - \sum_i   \left( \sum_{l \neq i} v_{l}^\top x_{l}^{(k^*(v))} + \boost^{(k^*(v))} \right).
    \end{align*}
    Since $k^*(v)$ is fixed for different $\lambda$ and, by linear combination, it holds that ${\hat{x}}^{(k^*(\cdot))} = \lambda {\hat{x}}^{(1, k^*(\cdot))} + (1 - \lambda) {\hat{x}}^{(2, k^*(\cdot))}$ and ${\hat{\boost}}^{(k^*(\cdot))} = \lambda {\hat{\boost}}^{(1, k^*(\cdot))} + (1 - \lambda) {\hat{\boost}}^{(2, k^*(\cdot))}$, 
    \begin{align*}
        &\sum_i    \left( \sum_{l \neq i} v_{l}^\top x_{l}^{(k^*(v_{-i}))} + \boost^{(k^*(v_{-i}))} \right)  - \sum_i   \left( \sum_{l \neq i} v_{l}^\top x_{l}^{(k^*(v))} + \boost^{(k^*(v))} \right) \\
        &~~~~~~ = \lambda \left[ \sum_i    \left( \sum_{l \neq i} v_{l}^\top {\hat{x}}_{l}^{(1, k^*(v_{-i}))} + {\hat{\boost}}^{(1, k^*(v_{-i}))} \right) \right.\\
        &~~~~~~~~~~~~~~~~~~~~~~~~~~~~~~~~~~~~~~~~~~~~~~~~~~~~~~~~~\left.- \sum_i   \left( \sum_{l \neq i} v_{l}^\top {\hat{x}}_{l}^{(1, k^*(v))} + {\hat{\boost}}^{(1, k^*(v))} \right)\right] \\
        &~~~~~~~~~+ (1 - \lambda) \left[\sum_i    \left( \sum_{l \neq i} v_{l}^\top {\hat{x}}_{l}^{(2, k^*(v_{-i}))} + {\hat{\boost}}^{(2, k^*(v_{-i}))} \right) \right.\\
        &~~~~~~~~~~~~~~~~~~~~~~~~~~~~~~~~~~~~~~~~~~~~~~~~~~~~~~~~~\left.- \sum_i   \left( \sum_{l \neq i} v_{l}^\top {\hat{x}}_{l}^{(2, k^*(v))} + {\hat{\boost}}^{(2, k^*(v))} \right)\right] \\
        &~~~~~~= \lambda \rev{ \widehat M_1} + (1 - \lambda)  \rev{ \widehat M_2}.
    \end{align*}This completes the proof.  
\end{proof}

\subsection{Discretizing large menus} \label{sec:AMA-discrete}
This subsection provides a detailed proof of Theorem~\ref{thm:AMA-2}. To do this, we will show that, for an AMA with a large number of options, one can discretize it such that, after discretization, the menu is $0$-reducible. Additionally, during this discretization, the revenue loss will be up to $\varepsilon$.

\begin{lemma}\label{lem:AMA-2-support}
    Consider an  AMA $M_1$ with at least $K + 1 = \yys$ options. There exists an $0$-reducible menu $\widetilde{M}_1$, such that, for any linear combination of $M_1$ and $\widetilde M_1$, $M = \lambda M_1 + (1 - \lambda) \widetilde{M}_1$ for $\lambda \in [0, 1]$, $\rev{ M} \geq \rev{M_1} - \varepsilon$.  
\end{lemma}
Theorem~\ref{thm:AMA-2} simply follows by combining Lemma~\ref{lem:AMA-2-support} and Proposition~\ref{prop:AMA-zero}.

Note that the payments and allocations only depend on those boosted welfare maximizing options. Therefore, to show that $\widetilde M_1$ is $0$-reducible, it suffices to show $\widetilde M_1$ has at most $\sqrt{K+1}$ different allocations.

We now formally define $\widetilde{M}_1$. We introduce parameters $\tilde \varepsilon$ and $\delta$, which will be specified later.

\paragraph{Construction of $\widetilde{M}_1$}  For $x^{(k)}$, we round it to $\tilde{x}^{(k)}$ in which $\tilde{x}^{(k)}_{ij} = \frac{\tilde \varepsilon}{m} \left\lfloor \frac{m x_{ij}^{(k)}}{\tilde \varepsilon} \right \rfloor$. With this rounding, for any $v \in [0, 1]^{nm}$ such that $\|v_i\| \leq 1$ for $i$, \footnote{The second inequality holds as $v_i^\top( x_i^{(k)} - \tilde{x}_i^{(k)}) \leq \|v_i\|_1 \|x_i^{(k)} - \tilde{x}_i^{(k)}\|_{\infty} \leq \frac{\tilde \varepsilon}{m}$ for any $i$.}
\begin{align*}
    \sum_i v^\top_i x_i^{(k)} \geq \sum_i v^\top_i \tilde{x}_i^{(k)} \geq \sum_i v^\top_i x_i^{(k)} - \tilde \varepsilon. \numberthis \label{AMA:eps-cover} 
\end{align*}  For $\boost^{(k)}$, we let $\tilde{\boost}^{(k)} = (1 - \delta) \boost^{(k)}$. 

\begin{lemma}
    For any given $0 < \varepsilon \leq \frac{1}{4}$, let $\delta = \frac{\sqrt{\tilde \varepsilon}}{m}$ and $\tilde \varepsilon = \frac{\varepsilon^2}{16m^2}$. Then, $$\rev{\widetilde M_1} \geq \rev{M_1} - \varepsilon.$$ The number of different allocations in $\widetilde M_1$ is at most  $\lceil \frac{16m^3}{\varepsilon^2} \rceil^{nm}$.  Additionally, for any linear combination of $M_1$ and $\widetilde M_1$, $M = \lambda M_1 + (1 - \lambda) \widetilde{M}_1$, $\rev{ M} \geq \rev{M_1} - \varepsilon$.  
\end{lemma}
    \begin{proof} We first demonstrate $\rev{\widetilde M_1} \geq \rev{M_1} - m \tilde{\varepsilon} - \frac{m^2 \delta}{1 - \delta} - \frac{\tilde{\varepsilon}}{\delta}$. The result follows by picking $\delta = \frac{\sqrt{\tilde \varepsilon}}{m}$ and $\tilde \varepsilon = \frac{\varepsilon^2}{16m^2}$. The proof of the bound on the linear combination of  $M_1$ and $\widetilde{M_1}$ is analogous. We use the notation $M_1 = \{ x^{(k)}, \boost^{(k)}\}_{k\in\mathcal{K}}$ and  $\widetilde M_1 = \{ \tilde{x}^{(k)}, \tilde{\boost}^{(k)}\}_{k\in\mathcal{K}}$.
    
    We fix the valuation $v$. Let $k(v) = \arg \max_{k} \sum_i  v_i^{\top} x_{i}^{(k)} + \boost^{(k)}$ and satisfy the tie-breaking rule.  The total payment using $M_1$ can be expressed as follows:
\begin{align*}
    &\sum_i \left(\sum_{l \neq i} v_l^{\top} x_{l}^{(k(v_{-i}))} + \boost^{(k(v_{-i}))}\right) -  \sum_i \left(\sum_{l \neq i} v_l^{\top} x_{l}^{(k(v))} - \boost^{(k(v))}\right)  \\
    &~~~ = \underbrace{\sum_i \left(\sum_{l \neq i} v_l^{\top} x_{l}^{(k(v_{-i}))} + \boost^{(k(v_{-i}))}\right)}_{A_1} \underbrace{- (m - 1) \left(\sum_i v_i^{\top} x_{i}^{(k(v))} + \boost^{(k(v))} \right)}_{B_1} - \boost^{(k(v))}.\numberthis \label{ineq:AMA-disc-0}
\end{align*}
Similarly, let $\tilde{k}(v)  = \arg \max_{k} \sum_i  v_i^{\top} \tilde{x}_{i}^{(k)} + \tilde{\boost}^{(k)}$, then,  the total payment with $\widetilde M_1$ is 
\begin{align*}
    \underbrace{\sum_i \left(\sum_{l \neq i} v_l^{\top} \tilde{x}_{l}^{(\tilde{k}(v_{-i}))} + \tilde{\boost}^{(\tilde{k}(v_{-i}))}\right)}_{A_2} \underbrace{- (m - 1) \left(\sum_i v_i^{\top} \tilde{x}_{i}^{(\tilde{k}(v))} + \tilde{\boost}^{(\tilde{k}(v))} \right)}_{B_2} - \tilde{\boost}^{(\tilde{k}(v))}. \numberthis \label{ineq:AMA-disc-1}
\end{align*}
We bound the differences between $A_1$ and $A_2$, $B_1$ and $B_2$, and $\boost^{(k(v))}$ and $\tilde{\boost}^{(\tilde{k}(v))}$ separately. 

First, for the difference between $A_1$ and  $A_2$, we can use the following inequalities: for any possible $v$ such that $\|v_i\| \leq 1$ for all $i$,
    \begin{align*}
          \sum_i v_i^{\top} x_{i}^{(k(v))} + \boost^{(k(v))}  &\leq \sum_i v_i^{\top} \tilde{x}_{i}^{(k(v))} + \tilde{\boost}^{(k(v))} + \tilde{\varepsilon} + \delta \boost^{(k(v))} \\
          &\leq \sum_i v_i^{\top} \tilde{x}_{i}^{(\tilde{k}(v))} + \tilde{\boost}^{(\tilde{k}(v))} + \tilde{\varepsilon} + \delta \boost^{(k(v))}.
    \end{align*}
    This implies 
\begin{align*}
    &A_1 \leq A_2 + m \tilde{\varepsilon} +  \sum_i \delta \boost^{(k(v_{-i}))}. \numberthis \label{ineq:AMA-disc-2}
\end{align*}

Second, for the difference between $B_1$ and $B_2$, we can use the following inequalities: for any possible $v$ such that $\|v_i\| \leq 1$ for all $i$,
    \begin{align*}
    \sum_i v_i^{\top} x_{i}^{(k(v))} + \boost^{(k(v))}  &\geq \sum_i v_i^{\top} x_{i}^{(\tilde{k}(v))} + \boost^{(\tilde{k}(v))} \\
    &\geq  \sum_i v_i^{\top} \tilde{x}_{i}^{(\tilde{k}(v))} + \tilde{\boost}^{(\tilde{k}(v))} + \delta \boost^{(\tilde{k}(v))}.
    \end{align*}
    This implies 
    \begin{align*}
        &B_1\leq B_2 - (m-1) \delta \boost^{(\tilde{k}(v))}. \numberthis \label{ineq:AMA-disc-3}
    \end{align*}

Finally,  we want to claim
\begin{align*}
    \boost^{(\tilde{k}(v))} \leq \boost^{(k(v))} + \frac{\tilde{\varepsilon}}{\delta}, \text{ which implies } {\tilde{\boost}}^{(\tilde{k}(v))} \leq \boost^{(k(v))} + \frac{\tilde{\varepsilon}}{\delta}  - \delta \boost^{(\tilde{k}(v))};\numberthis \label{ineq:AMA-disc-4}
\end{align*} as otherwise if $\boost^{(\tilde{k}(v))} >  \boost^{(k(v))} +  \frac{\tilde{\varepsilon}}{\delta}$, then this implies
\begin{align*}
    \sum_i v_i^{\top} \tilde{x}^{(k(v))}_{i} + \tilde{\boost}^{(k(v))} &\geq  \sum_i  v_i^{\top} x^{(k(v))}_{i} + \boost^{(k(v))} - \tilde{\varepsilon} - \delta \boost^{(k(v))}\\
    &\geq \sum_i  v_i^{\top} x^{(\tilde{k}(v))}_{i} + \boost^{(\tilde{k}(v))} - \tilde{\varepsilon} - \delta \boost^{(k(v))}\\
    &\geq  \sum_i  v_i^{\top} \tilde{x}^{(\tilde{k}(v))}_{i} + \tilde{\boost}^{(\tilde{k}(v))} - \tilde{\varepsilon} + \delta (\boost^{(\tilde{k}(v))} - \boost^{(k(v))})\\
    &> \sum_i  v_i^{\top} \tilde{x}^{(\tilde{k}(v))}_{i} + \tilde{\boost}^{(\tilde{k}(v))}, \numberthis \label{ineq:AMA-discrete-3}
\end{align*}
which contradicts the fact that $\tilde{k}(v)  = \arg \max_{k} \sum_i v_i^{\top} \tilde{x}_{i}^{(k)} + \tilde{\boost}^{(k)}$.

By combining the formula of total payment with $M_1$, \eqref{ineq:AMA-disc-0}, the formula of total payment with $\widetilde M_1$, \eqref{ineq:AMA-disc-1}, and inequalities \eqref{ineq:AMA-disc-2}, \eqref{ineq:AMA-disc-3}, \eqref{ineq:AMA-disc-4}; the loss on the total payment is at most $m \tilde{\varepsilon} +  \sum_i \delta \boost^{(k(v_{-i}))} - m  \delta \boost^{(\tilde{k}(v))} + \frac{\tilde{\varepsilon}}{\delta} $. Note that $\tilde{\boost}^{(k(v_{-i}))} \leq \tilde{\boost}^{(\tilde{k}(v))}  + m$\footnote{This is true becasue $v^\top \tilde{x}^{(k(v_{-i}))} +\tilde{\boost}^{(k(v_{-i}))} \leq  v^\top \tilde{x}^{(\tilde{k}(v))} +\tilde{\boost}^{(\tilde{k}(v))} $ and $v^\top \tilde{x}^{(k(v))} \leq m$.}. Therefore, the total loss on the payment is at most $m \tilde{\varepsilon} + \frac{m^2 \delta}{1 - \delta} + \frac{\tilde{\varepsilon}}{\delta} $ which is $\rev{\widetilde M_1} \geq \rev{M_1} - m \tilde{\varepsilon} - \frac{m^2 \delta}{1 - \delta} - \frac{\tilde{\varepsilon}}{\delta}$.

Now, we prove a similar result for $M$, which is a linear combination of $M_1$ and $\widetilde M_1$: $M = \lambda M_1 + (1 - \lambda) \widetilde M_1$.  Let $M = ({x'}^{(k)}, {\boost'}^{(k)})_{k\in\mathcal{K}}$ and   $k'(v)  = \arg \max_{k} \sum_i  v_i^{\top} {x'}_{i}^{(k)} + {\boost'}^{(k)}$. The total payment of $M$ is 
\begin{align*}
    \underbrace{\sum_i \left(\sum_{l \neq i} v_l^{\top} {x'}_{l}^{(k(v_{-i}))} + {\boost'}^{(k(v_{-i}))}\right)}_{A_3} \underbrace{- (m - 1) \left(\sum_i v_i^{\top} {x'}_{i}^{(k(v))} + {\boost'}^{(k(v))} \right)}_{B_3} - {\boost'}^{(k(v))}. \numberthis \label{eq:total-discrete-M}
\end{align*}Note that,  by linear combination, for any $k$, $\sum_i v_i^\top x_i^{(k)} \geq \sum_i v_i^\top {x'}_i^{(k)} \geq \sum_i v_i^\top x_i^{(k)} - \tilde \varepsilon$ and ${\boost'}^{(k)} = (1 - \delta') \boost^{(k)}$ such that $\delta' = (1 - \lambda) \delta$.  With similar proofs as 
 \eqref{ineq:AMA-disc-2}  and \eqref{ineq:AMA-disc-3}, the following 
 two inequalities hold, 
\begin{align*}
    &A_1\leq A_3 + m \tilde{\varepsilon} +  \sum_i \delta' \boost^{(k(v_{-i}))}; \numberthis \label{ineq:AMA-disc-2-cb}\\
        &B_1 \leq B_3 - (m-1) \delta' \boost^{({k'}(v))}.  \numberthis \label{ineq:AMA-disc-3-cb}
\end{align*}
And, similarly, \begin{align*}
    \boost^{({k'}(v))} \leq \boost^{(k(v))} + \frac{\tilde{\varepsilon}}{\delta},  \numberthis \label{ineq:AMA-disc-4-cb}
\end{align*} as otherwise
 $\boost^{({k'}(v))} >  \boost^{(k(v))} +  \frac{\tilde{\varepsilon}}{\delta}$ implies
\begin{align*}
    &\sum_i v_i^{\top} {x'}^{(k(v))}_{i} + {\boost'}^{(k(v))}  \\
    &~~~~~~~~~= \lambda (\sum_i v_i^{\top} x^{(k(v))}_{i} + \boost^{(k(v))}) + (1 - \lambda) (\sum_i v_i^{\top} \tilde{x}^{(k(v))}_{i} + \tilde{\boost}^{(k(v))}) \\
    &~~~~~~~~~> \lambda (\sum_i v_i^{\top} x^{(\hat{k}(v))}_{i} + \boost^{(\hat{k}(v))}) + (1 - \lambda) (\sum_i v_i^{\top} \hat{x}^{(k(v))}_{i} + \hat{\boost}^{(k(v))})\\
    &~~~~~~~~~ = \sum_i v_i^{\top} {x'}^{({k'}(v))}_{i} + {\boost'}^{({k'}(v))}.
\end{align*}
The strict inequality follows from \eqref{ineq:AMA-discrete-3} and $\lambda < 1$.

Therefore, combining inequalities \eqref{ineq:AMA-disc-2-cb}, \eqref{ineq:AMA-disc-3-cb}, \eqref{ineq:AMA-disc-4-cb}, the total payment  with $M_1$, \eqref{ineq:AMA-disc-0}, and the total payment with $M$, \eqref{eq:total-discrete-M};
the total loss for $M$ is at most $m \tilde{\varepsilon} + \frac{m^2 \delta'}{1 - \delta'} + \frac{\tilde{\varepsilon}}{\delta}  \leq m \tilde{\varepsilon} + \frac{m^2 \delta}{1 - \delta} + \frac{\tilde{\varepsilon}}{\delta} $.

The result follows by picking $\delta = \frac{\sqrt{\tilde \varepsilon}}{m}$ and $\tilde \varepsilon = \frac{\varepsilon^2}{16m^2}$.
\end{proof}
\subsection{Difference between softmax and argmax in revenue}\label{sec:softmax-ama}
Similar to the case of the \emph{RochetNet}, in the training process, softmax operation is used instead of argmax. Recall that the revenue of AMA is 
\begin{align*}
    \texttt{Rev} = \int_v \sum_i p_i(v) \texttt{d} F(v),
\end{align*}
where
\begin{align*}
    p_i(v) = & \left( \sum_{l \neq i}  v^\top_l x_{l}^{(k(v_{-i}))} + \boost^{(k(v_{-i}))} \right)  -  \left( \sum_{l \neq i}  v^\top_l x_{l}^{(k(v))} + \boost^{(k(v))} \right),
\end{align*}
and 
\begin{align*}
    k(v) = \arg \max_{k} \sum_i  v^\top_i x_{i}^{(k)} + \boost^{(k)}.
\end{align*}
For the softmax version, instead of using $k(v)$ which exactly maximizes the boosted social welfare, now $k^{\texttt{softmax}}(v)$ is a random variable:
\begin{align*}
    k^{\texttt{softmax}}(v) = k \text{ with probability } \frac{e^{Y  (v^\top_i x_{i}^{(k)} + \boost^{(k)})}}{ \sum_{k'} e^{Y  (v^\top_i x_{i}^{(k')} + \boost^{(k')})}};
\end{align*}
 and the price is the expectation on $k(v)$
\begin{align*}
    p^{\texttt{softmax}}_i(v) = &\mathbb{E}\left[\left( \sum_{l \neq i} v^\top_l x_{l}^{(k^{\texttt{softmax}}(v_{-i}))} + \boost^{(k^{\texttt{softmax}}(v_{-i}))} \right)  \right. \\
    &~~~~~~~~~~~~~~~~~~~~~\left.-  \left( \sum_{l \neq i} v^\top_l x_{l}^{(k^{\texttt{softmax}}(v))} + \boost^{(k^{\texttt{softmax}}(v))} \right)\right];
\end{align*}
and the revenue is
\begin{align*}
    \revsoftmax{M} = \int_v \sum_i p^{\texttt{softmax}}_i(v) \texttt{d} F(v).
\end{align*}
We show the following result. Note that we also assume the maximal density of a valuation type is $\mathcal{X}$.
\begin{theorem} \label{thm:softmax-AMA}
\begin{align*}
    |\revsoftmax{M} - \rev{M}| \leq \frac{
 m(K + 1)}{eY} + \frac{n  m \mathcal{X}(K + 1)}{Y} \left(1 + \log \frac{mY}{ m\mathcal{X} } \right).
    \end{align*}
\end{theorem}
To prove this theorem, we need the following lemma which provides one of the basic properties of the softmax.
\begin{lemma}\label{lem:base-prop-softmax}
    Given $L$ values, $a_1 \geq a_2 \geq a_3 \geq \cdots \geq a_L$, then $0 \leq a_1 - \sum_k a_k \frac{e^{Y a_k}}{\sum_{k'} e^{Y a_{k'}}} \leq \frac{L}{e Y}$.
\end{lemma}
\begin{proof}
    It's clear that $0 \leq a_1 - \sum_k a_k \frac{e^{Y a_k}}{\sum_{k'} e^{Y a_{k'}}} $. On the other direction, 
\begin{align*}
    a_1 - \sum_k a_k \frac{e^{Y a_k}}{\sum_{k'} e^{Y a_{k'}}} &\leq \sum_k (a_1 - a_k) \frac{e^{Y (a_k - a_1)}}{\sum_{k'} e^{Y (a_{k'} - a_1)}} \\
    &\leq \frac{1}{Y} \sum_k Y (a_1 - a_k) \frac{e^{Y (a_k - a_1)}}{\sum_{k'} e^{Y (a_{k'} - a_1)}} \\
    &\leq \frac{1}{Y} \sum_k e^{Y (a_1 - a_k) - 1} \frac{e^{Y (a_k - a_1)}}{\sum_{k'} e^{Y (a_{k'} - a_1)}} \\ 
    &\leq \frac{L}{e Y} \frac{1}{\sum_{k'} e^{Y (a_{k'} - a_1)}} \\
    &\leq \frac{L}{e Y}.\tag*{\qedhere}
\end{align*}
\end{proof}
Now, we can prove Theorem~\ref{thm:softmax-AMA}.

\begin{proof}[Proof of Theorem~\ref{thm:softmax-AMA}]
    We first give the upper bound on $\rev{M} - \revsoftmax{M}$. 
    
    Let $k(v) = \arg \max_{k} \sum_i  v^\top_i x_{i}^{(k)} + \boost^{(k)}$ be the rule used in $\rev{M}$. Recall that 
    \begin{align*}
        \sum_i p_i(v) &= \sum_i  \left(\sum_{l \neq i} v^\top_l x_{l}^{(k(v_{-i}))} + \boost^{(k(v_{-i}))}\right) - \sum_i   \left(\sum_{l \neq i}   v^\top_l x_{l}^{(k(v))} + \boost^{(k(v))} \right)
    \end{align*}
    and
    \begin{align*}
        \sum_i p^{\texttt{softmax}}_i(v) &= \mathbb{E} \left[\sum_i   \left(\sum_{l \neq i} v^\top_l x_{l}^{(k^{\texttt{softmax}}(v_{-i}))} + \boost^{(k^{\texttt{softmax}}(v_{-i}))}\right)\right.\\
        &~~~~~~~~~\left.- \sum_i    \left(\sum_{l \neq i} v^\top_l x_{l}^{(k^{\texttt{softmax}}(v))} + \boost^{(k^{\texttt{softmax}}(v))} \right)  \right].
    \end{align*}
    Then, 
    \begin{align*}
        &\rev{M} - \revsoftmax{M} \\
        &~~~~~~= \mathbb{E}_v \left[\sum_i   \left(\sum_{l \neq i} v^\top_l x_{l}^{(k(v_{-i}))} + \boost^{(k(v_{-i}))}\right) \right. \\
        &~~~~~~~~~~~~~~~~~~~~~~~~~~~~~~~~~~~~~~~~~~~~-  \mathbb{E} \left[\sum_i   \left(\sum_{l \neq i} v^\top_l x_{l}^{(k^{\texttt{softmax}}(v_{-i}))} + \boost^{(k^{\texttt{softmax}}(v_{-i}))}\right)\right]  \\
        &~~~~~~~~~~~~~~~~~~~~~~~~~~~-  m\left(\sum_i v^\top_i x_{i}^{(k(v))} + \boost^{(k(v))} \right) \\
        &~~~~~~~~~~~~~~~~~~~~~~~~~~~~~~~~~~~~~~~~~~~~+\mathbb{E}\left[ m \left(\sum_i v^\top_i x_{i}^{(k^{\texttt{softmax}}(v))} + \boost^{(k^{\texttt{softmax}}(v))} \right)\right] \\
        &~~~~~~~~~~~~~~~~~~~~~~~~~~~\left.+ \sum_i v^\top_i x_{i}^{(k(v))} -  \mathbb{E}\left[ \sum_i v^\top_i x_{i}^{(k^{\texttt{softmax}}(v))}\right] \right] \\
        &~~~~~~\leq \frac{ m(K + 1)}{eY} + \mathbb{E}\left[ \sum_i v^\top_i x_{i}^{(k(v))}\right] -  \mathbb{E}\left[ \sum_i v^\top_i x_{i}^{(k^{\texttt{softmax}}(v))}\right].
    \end{align*}
    The inequality follows by Lemma~\ref{lem:base-prop-softmax}. Then, we will bound the difference between $\mathbb{E}\left[ \sum_i v^\top_i x_{i}^{(k(v))}\right]$ and $\mathbb{E}\left[ \sum_i v^\top_i x_{i}^{(k^{\texttt{softmax}}(v))}\right]$. More specifically, 
    \begin{align*}
         &\mathbb{E}[v_{ij} x_{ij}^{(k(v))}] - \mathbb{E}[ v_{ij} x_{ij}^{(k^{\texttt{softmax}}(v))}] \\
         &~~~~~~~~~= \int_v v_{ij} x_{ij}^{(k(v))} - \sum_k  v_{ij} x_{ij}^{(k)} \frac{e^{Y  (\sum_{i'}   v^\top_{i'} x_{i'}^{(k)} + \boost^{(k)})}}{ \sum_{k'} e^{Y  (\sum_{i'}   v^\top_{i'} x_{i'}^{(k')} + \boost^{(k')})}} \texttt{d} F(v) \\
        &~~~~~~~~~\leq  \int_v \sum_{k}  (x_{ij}^{(k)} -x_{ij}^{(k(v))})^{+}  \frac{e^{Y  (\sum_{i'}  v^\top_{i'} x_{i'}^{(k)} + \boost^{(k)})}}{ \sum_{k'} e^{Y  (\sum_{i'}  v^\top_{i'} x_{i'}^{(k')} + \boost^{(k')})}}\texttt{d} F(v)  \\
        &~~~~~~~~~\leq  \int_v \sum_{k}  (x_{ij}^{(k)} -x_{ij}^{(k(v))})^{+}  e^{Y  (\sum_{i'}  v^\top_{i'} x_{i'}^{(k)} + \boost^{(k)}- \sum_{i'}  v^\top_{i'} x_{i'}^{(k(v))} - \boost^{(k(v))})}\texttt{d} F(v).
    \end{align*}
Recall that $(\cdot)^+ \triangleq \max\{ \cdot, 0\}$. Let's define $\BW{k} = \sum_{i'}  v^\top_{i'} x_{i'}^{(k)} + \boost^{(k)}$ to be the boosted welfare of option $k$ for simplicity. Now, we focus on one option $k$, and we will give an upper bound on 
\begin{align*}
    \int (x_{ij}^{(k)} -x_{ij}^{(k(v))})^{+} \mathsf{1}_{\BW{k} + \sigma \geq \BW{k(v)} \geq \BW{k}  } \text{d} F(v) \numberthis \label{eq:AMA-upper-diff-rev}
\end{align*}
for the non-negative $\sigma$. The value of $\sigma$ will be determined later. Note that it is always true that $\BW{k(v)} \geq \BW{k}$ by the definition of $k(v)$. Additionally, if $\BW{k} + \sigma \geq \BW{k(v)} $ is not satisfied then $e^{Y  (\sum_{i'} v^\top_{i'} x_{i'}^{(k)} + \boost^{(k)} - \sum_{i'}  v^\top_{i'} x_{i'}^{(k(v))} - \boost^{(k(v))})}\leq e^{- Y \sigma}$. Therefore, if \eqref{eq:AMA-upper-diff-rev} is upper bounded by $\mathcal{C}_{ij}(\sigma)$, then $\revsoftmax{M} - \rev{M} \leq  \frac{ mK}{eY} + (K + 1) \left(\sum_{ij} \mathcal{C}_{ij}(\sigma) +  nm e^{- Y \sigma}\right)$. 

Note that 
\begin{align*}
    & \int (x_{ij}^{(k)} -x_{ij}^{(k(v))})^{+} \mathsf{1}_{\BW{k} + \sigma \geq \BW{k(v)} \geq \BW{k}  } \text{d} F(v)\\
    &~~~~~~=   \int   (x_{ij}^{(k)} -x_{ij}^{(k(v))})^{+} \mathsf{1}_{\BW{k(v)} - \BW{K} \in [0, \sigma]  } \text{d} F(v). \numberthis 
    \label{ineq:AMA-softmax-os}
\end{align*}
Now we fix all coordinates of valuation $v$ other than coordinate $ij$. Note that, the function $ \BW{k(v)} - \BW{k} = \sum_{i'} v^\top_{i'} {x_{i'}^{(k(v))}} + \boost^{(k(v))} -  \sum_{i'}  v^\top_{i'} {x_{i'}^{(k)}} - \boost^{(k)}$ is a convex function on $v_{ij}$ and $ (x^{(k(v))}_{ij} - x^{(k)}_{ij})$ is the gradient. Since we are looking at the region such that the function  $\BW{k(v)} - \BW{k}$ is bounded in $[0, \sigma]$, this direct implies
\begin{align*}
  \int (x_{ij}^{(k)} -x_{ij}^{(k(v))})^{+} \mathsf{1}_{\BW{k} + \sigma \geq \BW{k(v)} \geq \BW{k}  } \text{d} F(v)  \leq    \mathcal{X}\sigma.
\end{align*}
as the maximal density is at most $\mathcal{X}$.
This implies $\rev{M} - \revsoftmax{M} \leq  \frac{
 m(K + 1)}{eY} + (K + 1)n  m \mathcal{X} \sigma + (K + 1)nm e^{- Y \sigma}$ which is upper bounded by $\frac{
 m(K + 1)}{eY} + \frac{n  m \mathcal{X}(K + 1)}{Y} \left(1 + \log \frac{Y}{ \mathcal{X} } \right)$ by setting $\sigma = \frac{1}{Y} \log \frac{Y}{  \mathcal{X}}$.

The upper bound on $\revsoftmax{M} - \rev{M}$ follows by a similar argument.
\end{proof}

\end{document}